\documentclass[11pt,a4paper,reqno,intlimits,sumlimits]{amsart}

\usepackage{amssymb}
\usepackage{color}
\usepackage[mathscr]{euscript}
\usepackage{xspace}
\usepackage{enumerate}
\usepackage{enumitem}
\usepackage{epsfig,rotating}
\usepackage{graphicx}
\input{xy}\xyoption{all}
\usepackage{todonotes}
\usepackage{amsmath}
\usepackage{dsfont}
\usepackage{color}
\usepackage{soul}
\usepackage[longnamesfirst,round]{natbib}

\usepackage[left=1.55in,right=1.55in,top=1.6in,bottom=1.3in]{geometry}

\DeclareMathAlphabet{\mathpzc}{OT1}{pzc}{m}{it}

\usepackage[bookmarksopen,pdfstartview=FitH]{hyperref}
\hypersetup{colorlinks,%
           citecolor=blue,%
           filecolor=blue,%
           linkcolor=blue,%
           urlcolor=blue}

\begin{document}

\theoremstyle{plain}
\newtheorem{theorem}{Theorem}[section]
\newtheorem{lemma}[theorem]{Lemma}
\newtheorem{proposition}[theorem]{Proposition}
\newtheorem{corollary}[theorem]{Corollary}
\newtheorem{claim}[theorem]{Claim}
\newtheorem{definition}[theorem]{Definition}

\theoremstyle{definition}
\newtheorem{remark}[theorem]{Remark}
\newtheorem{note}[theorem]{Note}
\newtheorem{example}[theorem]{Example}
\newtheorem{assumption}[theorem]{Assumption}
\newtheorem*{notation}{Notation}
\newtheorem*{assuL}{Assumption ($\mathbb{L}$)}
\newtheorem*{assuAC}{Assumption ($\mathbb{AC}$)}
\newtheorem*{assuEM}{Assumption ($\mathbb{EM}$)}
\newtheorem*{assuES}{Assumption ($\mathbb{ES}$)}
\newtheorem*{assuM}{Assumption ($\mathbb{M}$)}
\newtheorem*{assuMM}{Assumption ($\mathbb{M}'$)}
\newtheorem*{assuL1}{Assumption ($\mathbb{L}1$)}
\newtheorem*{assuL2}{Assumption ($\mathbb{L}2$)}
\newtheorem*{assuL3}{Assumption ($\mathbb{L}3$)}

\newcommand{\Law}{\ensuremath{\mathop{\mathrm{Law}}}}
\newcommand{\loc}{{\mathrm{loc}}}
\newcommand{\Log}{\ensuremath{\mathop{\mathcal{L}\mathrm{og}}}}
\newcommand{\Meixner}{\ensuremath{\mathop{\mathrm{Meixner}}}}

\let\SETMINUS\setminus
\renewcommand{\setminus}{\backslash}

\def\stackrelboth#1#2#3{\mathrel{\mathop{#2}\limits^{#1}_{#3}}}

\def\blue{\color{blue}}
\def\red{\color{red}}

\renewcommand{\theequation}{\thesection.\arabic{equation}}
\numberwithin{equation}{section}

\newcommand\llambda{{\mathchoice
      {\lambda\mkern-4.5mu{\raisebox{.4ex}{\scriptsize$\backslash$}}}
      {\lambda\mkern-4.83mu{\raisebox{.4ex}{\scriptsize$\backslash$}}}
      
{\lambda\mkern-4.5mu{\raisebox{.2ex}{\footnotesize$\scriptscriptstyle\backslash$
}}}
      
{\lambda\mkern-5.0mu{\raisebox{.2ex}{\tiny$\scriptscriptstyle\backslash$}}}}}

\newcommand{\prozess}[1][L]{{\ensuremath{#1=(#1_t)_{0\le t\le T}}}\xspace}
\newcommand{\prazess}[1][L]{{\ensuremath{#1=(#1_t)_{0\le t\le T^*}}}\xspace}

\newcommand{\lijepoa}{{\mathcal{A}}}
\newcommand{\lijepob}{{\mathcal{B}}}
\newcommand{\lijepoc}{{\mathcal{C}}} 
\newcommand{\lijepod}{{\mathcal{D}}} 
\newcommand{\lijepoe}{{\mathcal{E}}}
\newcommand{\lijepof}{{\mathcal{F}}}
\newcommand{\lijepog}{{\mathcal{G}}}
\newcommand{\lijepok}{{\mathcal{K}}}
\newcommand{\lijepoo}{{\mathcal{O}}}
\newcommand{\lijepop}{{\mathcal{P}}}
\newcommand{\lijepoh}{{\mathcal{H}}}
\newcommand{\lijepom}{{\mathcal{M}}}
\newcommand{\lijepor}{{\mathcal{R}}}
\newcommand{\lijepou}{{\mathcal{U}}}
\newcommand{\lijepov}{{\mathcal{V}}}
\newcommand{\lijepoy}{{\mathcal{Y}}}
\newcommand{\er}{{\mathbb{R}}}
\newcommand{\ce}{{\mathbb{C}}}
\newcommand{\erd}{{\mathbb{R}^{d}}}
\newcommand{\en}{{\mathbb{N}}}
\newcommand{\de}{{\mathrm{d}}}
\newcommand{\De}{{\mathrm{D}}}
\newcommand{\im}{{\mathrm{i}}}
\newcommand{\indik}{{1}}
\newcommand{\D}{{\mathbb{D}}}
\newcommand{\E}{{\mathbb{E}}}
\newcommand{\N}{{\mathbb{N}}}
\newcommand{\Q}{{\mathbb{Q}}}
\renewcommand{\P}{{\mathbb{P}}}
\newcommand{\dd}{\operatorname{d}\!}
\newcommand{\ii}{\operatorname{i}\kern -0.8pt}
\newcommand{\Var}{\operatorname{Var }\,}
\newcommand{\dt}{\operatorname{d}\!t}   
\newcommand{\ds}{\operatorname{d}\!s}   
\newcommand{\dy}{\operatorname{d}\!y }    
\newcommand{\du}{\operatorname{d}\!u}  
\newcommand{\dv}{\operatorname{d}\!v}   
\newcommand{\dx}{\operatorname{d}\!x}   
\newcommand{\dq}{\operatorname{d}\!q}   
\newcommand{\fpoint}{\frac{\de}{\de t} f}
\newcommand{\gpoint}{\frac{\de}{\de t} g}
\newcommand{\varthetapoint}{\frac{\de }{\de t} \vartheta}
\newcommand{\gammapoint}{\frac{\de }{\de t} \gamma} 

\def\EM{\ensuremath{(\mathbb{EM})}\xspace}
\def\mg{martingale\xspace}
\def\smmg{semimartingale\xspace}
\def\smmgs{semimartingales\xspace}

\def\lmm{LIBOR market model\xspace}
\def\fpm{forward price model\xspace}
\def\mfm{Markov-functional model\xspace}
\def\alm{affine LIBOR model\xspace}
\def\lmms{LIBOR market models\xspace}
\def\fpms{forward price models\xspace}
\def\mfms{Markov-functional models\xspace}
\def\alms{affine LIBOR models\xspace}

\def\inK{\ensuremath{{k\in\mathcal{K}}}}
\def\inKn{\ensuremath{{k\in\mathcal{K}\setminus\{N\}}}}

\def\SL{\ensuremath{\mathbb{SL}}\xspace}

\def\I{\mathcal{I}}

\def\e{\mathrm{e}}
\def\lib{LIBOR\xspace}
\def\lev{L\'evy\xspace}
\def\loc{\mathrm{loc}}

\newcommand{\la}{\langle}
\newcommand{\ra}{\rangle}

\newcommand{\Norml}[1]{%
{|}\kern-.25ex{|}\kern-.25ex{|}#1{|}\kern-.25ex{|}\kern-.25ex{|}}

\newcommand{\Lip}{$\mathit{Lip}$}
\newcommand{\Int}{$\mathit{Int}$}
\newcommand{\Drift}{$\mathit{Drift}$}

\renewcommand{\1}{\mathds{1}}
\def\R{\ensuremath{\mathbb{R}}}


\title{A unified view of LIBOR models}

\author[K. Glau]{Kathrin Glau}
\author[Z. Grbac]{Zorana Grbac}
\author[A. Papapantoleon]{Antonis Papapantoleon}

\address{Center for Mathematics, Technical University of Munich, Parkring 
	 11, 85748 Garching b. M\"unchen, Germany}
\email{kathrin.glau@tum.de}

\address{Laboratoire de Probabilit{\'e}s et Mod\`eles Al{\'e}atoires, 
	 Universit{\'e} Paris Diderot, 75205 Paris Cedex 13, France}
\email{grbac@math.univ-paris-diderot.fr}

\address{Institute of Mathematics, TU Berlin, Stra\ss e des 17. Juni 136,
         10623 Berlin, Germany}
\email{papapan@math.tu-berlin.de}

\thanks{Financial support from the PROCOPE project ``Financial markets in 
transition: mathematical models and challenges'' and the Europlace Institute of 
Finance project ``Post-crisis models for interest rate markets'' is gratefully 
acknowledged.}

\maketitle\frenchspacing\pagestyle{myheadings}

\begin{abstract}
We provide a unified framework for modeling LIBOR rates using general 
semimartingales as driving processes and generic functional forms to describe 
the evolution of the dynamics. We derive sufficient conditions for the model to 
be arbitrage-free which are easily verifiable, and for the LIBOR rates to be 
true martingales under the respective forward measures. We discuss when the 
conditions are also necessary and comment on further desirable properties such 
as those leading to analytical tractability and positivity of rates. This 
framework allows to consider several popular models in the literature, such as 
LIBOR market models driven by Brownian motion or jump processes, the L\'evy 
forward price model as well as the affine LIBOR model, under one umbrella. 
Moreover, we derive structural results about LIBOR models and show, in 
particular, that only models where the forward price is an exponentially affine 
function of the driving process preserve their structure under different 
forward measures.
\end{abstract}

\section{Introduction} 

The \lib and EURIBOR interest rates rank among the most important interest rates 
worldwide. They are determined on a daily basis by a panel of banks for a number 
of maturities, while LIBOR is also determined for several currencies. LIBOR and 
EURIBOR serve as underlying rates for an enormous amount of financial 
transactions. In 2012, the outstanding values of contracts with \lib as 
reference were estimated at roughly summing up to $300$ trillion USD; see 
\citet{Wheatley2012}. Therefore, the development of suitable policies and 
regulations for the fair calculation of \lib and EURIBOR, as well as of 
mathematical models for the fair evaluation of interest rate products, is 
essential for the financial industry and also serves the general public 
interest.  

The modeling of the dynamics of LIBOR and EURIBOR rates is a challenging task 
due to the high dimensionality of the modeled objects. The major difference, 
from a modeling point of view, between interest rates and stock prices lies in 
the fact that stock prices are observed at each time point as a single value, 
once the difference between bid and ask prices is ignored, while interest rates 
are observed at each time point for several maturities. Moreover, these 
different rates (for the different maturities) are interdependent. Their joint 
modeling is indispensable because they jointly enter already the basic interest 
rate derivatives as underlying rates. In addition, the rates for different 
period lengths can no longer be derived from simple no-arbitrage relations. 
Indeed, the financial crisis of 2007--2009 has fundamentally changed the 
attitude of market participants towards risks in the interbank sector, regarding 
in particular counterparty and liquidity risk, which have a direct effect on the 
\lib rates for different lending periods; see, for example, 
\citet{Filipovic_Trolle_2011}. Summarizing, \lib modeling presents a challenge 
to jointly model the rates for different maturities and periods in an 
arbitrage-free way and such that the resulting pricing formulas are fast and 
accurately computable for all liquid derivatives, such as caps and swaptions. 
In this work, we provide a unified mathematical foundation for some of the most 
important of the existing LIBOR models in the literature. On this basis we gain 
valuable structural insight in modeling \lib rates. In particular we derive 
sufficient conditions for the validity of mandatory model features such as 
arbitrage-freeness and investigate those that typically support computational 
ease.

The seminalarticles by \citet{BraceGatarekMusiela97} and Miltersen, Sandmann, and Sondermann
\citeyearpar{MiltersenSandmannSondermann97} introduced the \lib Market Model (LMM), 
that became known also as the BGM model. The celebrity of the model certainly 
is, at least partly, owed to the fact that the BGM model reproduces the market 
standard Black's formula for caps. Moreover, the backward construction of \lib 
rates in the LMM presented in \citet{MusielaRutkowski97b} has proven to be 
extendible beyond models driven by Brownian motion. The article by 
\citet{EberleinOezkan05} introduced one of the first \lib models driven by jump 
processes and also proposed the L\'evy-driven forward price model. In 
\citeauthor{Jamshidian97} \citeyearpar{Jamshidian97,Jamshidian99} \lib models 
driven by general semimartingales were presented. In interest rate modeling, 
jump processes have several advantages. Firstly, just as in stock price 
modeling, their distributional flexibility allows to better capture the 
empirical distributions of logarithmic returns, see for instance 
\citeauthor{EberleinKluge04} \citeyearpar{EberleinKluge04,EberleinKluge05}. 
Secondly, the traditional way to jointly model the rates for different 
maturities would suggest to introduce one component of a multi-dimensional 
Brownian motion for each maturity, so as to introduce one stochastic factor for 
each source of risk. A L\'evy process with an infinite jump activity, in 
contrast, introduces infinitely many sources of risk, already as a 
one-dimensional process. In this regard, jump processes show their potential to 
reduce the dimension of the related computational problems for pricing and 
hedging. However, they also bring along a new level of technical challenges, in 
particular the measure changes between forward measures become more involved and 
the backward construction typically requires a more sophisticated justification. 
Additionally, various extensions of the \lib market model to stochastic 
volatility have also appeared in the literature, cf. \citet{Wu_Zhang_2006}, 
\citet{BelomestnyMathewSchoenmakers06} and 
\citet{Ladkau_Schoenmakers_Zhang_2013}. Recently, a modeling approach under one 
single forward measure, the terminal measure, has been proposed in 
\citet{KellerResselPapapantoleonTeichmann09}, which is based on affine 
processes. We refer to \citet{Schoenmakers05} and \citet{Papapantoleon10b} for 
an overview of the modeling approaches and the existing literature. Regarding 
the post-crisis LIBOR models we refer to \citet{BianchettiMorini13} and 
\citet{GrbacRunggaldier15}.

In view of the high level of technical sophistication that \lib models have 
reached in today's literature and also of the new demands they are faced with, 
we propose an abstract perspective on \lib modeling in order to obtain:
\begin{itemize}
\item a unified view on different modeling approaches, such as the \lib 
      market models, the L\'evy forward price models and the affine \lib models;
\item transparent conditions that guarantee:
\begin{itemize}
\item[$\circ$] positivity of bond prices and arbitrage-freeness -- the 
	       fundamental model requirements;
\item[$\circ$] martingality of the forward prices under their corresponding 
	       forward measures, which paves the way for change of numeraire 
	       techniques and tractable pricing formulas;
\item[$\circ$] structure preservation under different forward measures, a 
	       feature that is beneficial in connection with change of 
	       numeraire techniques;
\end{itemize}
\item the validity of further desirable model properties that lead to 
      analytically tractable models.
\end{itemize}

This article is structured as follows: in Section \ref{sec:2} we introduce the 
main modeling objects and formalize model axioms as well as desirable model 
properties that entail computational tractability. In Section \ref{sec:3}, we 
provide two general modeling approaches based on general semimartingales and 
generic functional forms for the evolution of rates, and derive sufficient 
conditions for the arbitrage-freeness of the models and for the forward price 
processes to be uniformly integrable martingales under their corresponding 
forward measures; positivity of bond prices holds by construction. On this basis 
we derive conditions that imply positivity of \lib rates and  ensure 
computational tractability. As an interesting additional insight we show that 
essentially only models in which the forward price processes are exponentials of 
an affine function of a semimartingale are structure preserving under different 
forward measures. In Section \ref{sec:4}, we present several \lib models in the 
guise of the general modeling framework and investigate sufficient conditions 
that lead to further essential model features. Finally, required results from 
semimartingale theory are derived in the appendix. 
\section{Axioms and desirable properties} 
\label{sec:2}

Let 
$(\Omega,\lijepof=\lijepof_{T_*},\mathbb{F}=(\lijepof_{t})_{t\in[0,T_*]},\P_N)$ 
denote a complete stochastic basis in the sense of 
\citet[Def.~I.1.3]{JacodShiryaev03}, where $T_*$ denotes a finite time horizon. 
Consider a discrete tenor structure $\mathcal{T}:=\{0=T_{0}<\ldots<T_{N}\le 
T_*\}$ with $\delta_k=T_{k}-T_{k-1}$ for $k\in\mathcal{K}:=\{1,\ldots,N\}$, and 
define $\bar{\mathcal{K}}:=\mathcal{K} \setminus \{N\}$. We assume that 
zero-coupon bonds with maturities $T_1, \ldots, T_N$ are traded in the market 
and denote by $B(t,T_k)$ the time-$t$ price of the zero-coupon bond with 
maturity $T_k$, for all $\inK$. We associate to each date $T_k$ the 
\textit{numeraire pair} $(B(\cdot,T_k),\P_k)$, meaning that bond prices 
discounted by the numeraire $B(\cdot,T_k)$ are $\P_k$-local martingales, for 
all $\inK$. The measures $\P_k$ are then called \textit{forward (martingale) 
measures}. Moreover, let $\mathcal{M}_\loc(\P)$ denote the set of local 
martingales with respect to the measure $\P$.

The \emph{forward \lib rate}, denoted by $L(t,T_{k})$, is a discretely 
compounded interest rate determined at time $t$ for the future accrual interval 
$[T_{k}, T_{k+1}]$. It is related to bond prices via 
\begin{equation}\label{eq:libor}
L(t, T_{k}) 
  = \frac{1}{\delta_{k}} \left( \frac{B(t, T_{k})}{B(t, T_{k+1})} - 1\right), 
  \quad \quad t \in [0, T_{k}],
\end{equation}
for $k \in \bar{\mathcal{K}}$. The \emph{forward price process} $F(\cdot, T_k, 
T_n)$ is defined as follows
\begin{equation}\label{eq:forward-price}
F(t, T_k, T_n)= \frac{B(t, T_k)}{B(t, T_n)}, 
	\quad \quad t \in[0, T_k \wedge T_n],
\end{equation}
for all $k, n\in \mathcal{K}$. The forward \lib rate $L(t, T_{k})$ and the 
forward price $F(t, T_{k}, T_{k+1})$ are connected via
\begin{equation}\label{eq:connection-Libor-forward-price}
F(t, T_{k}, T_{k+1}) = 1 + \delta_{k} L(t, T_{k}).
\end{equation}

We will describe in the sequel several axioms and properties that \lib models 
should posess in order to be economically meaningful on the one hand, and 
applicable in practice on the other. In particular, we will distinguish between 
three different groups of attributes. The first group consists of 
\textit{necessary axioms}, which are needed to build a sound financial model. 
These are:
\begin{enumerate}[label=$(\mathbb{A}1)$]
\item\label{a1} Bond prices are \textit{positive}, i.e. $B(\cdot,T_k)>0$ 
		for all $k\in\mathcal{K}$;
\end{enumerate}
\begin{enumerate}[label=$(\mathbb{A}2)$]
\item\label{a2} The model is \textit{arbitrage-free}, i.e. 
		$\frac{B(\cdot,T_k)}{B(\cdot,T_N)} \in \mathcal{M}_\loc(\P_N)$ 
		for all $k\in\mathcal{K}$.
\end{enumerate}
The first axiom is justified since bond prices are traded assets with a positive 
payoff, thus should have a positive price. The second axiom precludes the 
existence of arbitrage opportunities and could be equivalently formulated under 
any forward measure, i.e. the model is arbitrage-free if 
$\frac{B(\cdot,T_k)}{B(\cdot,T_n)} \in \mathcal{M}_\loc(\P_n)$ for all 
$k\in\mathcal{K}$ and some $n\in\mathcal{K}$; see also 
\citet[\S14.1.3]{MusielaRutkowski05} and \citet{Klein_Schmidt_Teichmann_2015}.

The second group consists of \textit{tractability properties}, which simplify 
computations in the model. Out of several possible choices, we will concentrate 
on the following: 
\begin{enumerate}[label=$(\mathbb{B}1)$]
\item\label{b1} Forward prices are \textit{true martingales}, i.e. 
		$\frac{B(\cdot,T_k)}{B(\cdot,T_N)}\in\mathcal{M}(\P_N)$ for all 
		$k\in\mathcal{K}$.
\end{enumerate}
\begin{enumerate}[label=$(\mathbb{B}2)$]
\item\label{b2} The model is \textit{structure preserving}, i.e. the \smmg
		characteristics of the driving process are transformed in a 
		deterministic way under forward measures.
\end{enumerate}
\begin{enumerate}[label=$(\mathbb{B}3)$]
\item\label{b3} Each \lib rate is a \textit{Markov process} under its 
		corresponding forward measure.
\end{enumerate}
\begin{enumerate}[label=$(\mathbb{B}4)$]
\item\label{b4} The initial \lib rates are direct model inputs.
\end{enumerate}
These properties are not necessary to build an arbitrage-free model, but are 
very convenient in several aspects. The first property allows to compute option 
prices as conditional expectations and to relate the forward measures via a 
density process. Hence, several option pricing formulas can be simplified 
considerably by changing to a more convenient forward measure. The second 
property yields that the processes driving each LIBOR rate remain in the same 
class of processes under each forward measure. Moreover, \ref{b1} combined with 
\ref{b2} typically allows to derive closed-form or semi-analytical pricing 
formulas for liquid products such as caps and swaptions. \ref{b3} also allows to 
simplify certain option pricing problems and use PDE methods. Additionally, if 
the initial term structure is a direct input in the model, i.e.\ \ref{b4} holds, 
then we avoid using a numerical procedure to fit the currently observed bond 
prices.

Finally, we shall also discuss the following model property:
\begin{enumerate}[label=$(\mathbb{C})$]
\item\label{c} \lib rates are always non-negative.
\end{enumerate}
Until the recent financial crisis, \lib rates were always non-negative, hence 
the possibility of rates becoming negative has been considered as a drawback of 
a model. As a consequence, several LIBOR models have been designed to produce 
non-negative \lib rates. Nowadays, the quoted \lib rates are at extremely low 
levels and even negative \lib rates for several tenors have been reported over 
longer time periods, which prompts us to take this into account in the modeling. 
Therefore, it is important to know which models allow for negative rates as well 
as which of the existing models for positive rates can easily be adapted to 
allow the rates to go below zero. Moreover, the techniques used to construct 
non-negative \lib rates can often be adapted to model other related 
non-negative quantities such as spreads in multiple curve models. 
\section{A unified construction of \lib models}
\label{sec:3}

Models for the evolution of \lib rates are constructed in the literature either 
using a backward induction approach, where rates are specified successively 
under different forward measures, or by modeling all rates simultaneously under 
one measure, typically the terminal forward measure. The former approach has 
been used for the construction of \lib market models and forward price models, 
while the latter is used for affine \lib models and Markov functional models. 
The aim of this section is to offer a unified construction of \lib models by 
emphasizing the common features in both approaches.

\subsection{Modeling rates via backward induction}
\label{sec:back-ind}

The aim of this subsection is to formulate sufficient conditions and to present 
a generic construction of \lib (market) models using the backward induction 
approach. The driving process is a general semimartingale and the functional 
form of the dynamics is also generic.

The following key observations of \citet{MusielaRutkowski97b} lie at the heart 
of the constructions via backward induction:
\begin{itemize}
\item A model for the \lib rates \((L(\cdot, T_k))_{k \in \bar{\mathcal{K}}}\) 
      is \textit{arbitrage-free} if \(L(\cdot, T_k)\) is a \(\P_{k+1}\)-local  
      martingale for all \(k \in \bar{\mathcal{K}}\).
\item The \textit{forward measures} \((\P_k)_{k \in \bar{\mathcal{K}}}\) are 
      related via the Radon-Nikodym derivatives
      \begin{align}\label{eq-def-PTk}
	\frac{\de \P_k}{\de \P_{k+1}} = \frac{1 + \delta_k L(T_k, T_k)}{1 + 
	\delta_k L(0, T_k)},\qquad \text{ for all } k \in \bar{\mathcal{K}}.
      \end{align}
\end{itemize}
Therefore, in order to construct a \lib model it suffices to specify the 
dynamics either of the \lib rate $L(\cdot, T_k)$ itself or of the forward price 
process $F(\cdot,T_k,T_{k+1})=1+\delta  L(\cdot,T_k)$ for all $k \in 
\bar{\mathcal{K}}$, and both choices determine the densities in 
\eqref{eq-def-PTk} as well.

Our construction is based on specifying an exponential semimartingale for the 
dynamics of the forward price process with the following functional form:
\begin{equation}\label{specify-F}
F(\cdot,T_k,T_{k+1}) = \e^{ f^k(\cdot,X) },
\end{equation}
where $f^k$ are functions for each $k \in \bar{\mathcal{K}}$ and $X$ is a 
semimartingale. This approach unifies the construction of the \lmms and the 
\fpms by appropriate choices of $f^k$ and $X$ that will be discussed in Section 
\ref{sec:3}.

Consider an $\R^{d}$-valued semimartingale $X=(X_{t})_{0\leq t \leq T_N}$ on 
$(\Omega,\lijepof,\mathbb{F},\P_N)$ and a collection of functions $f^{k}: 
[0,T_N] \times \erd \to \er$ for all $k \in \bar{\mathcal{K}}$, which satisfy 
the following assumptions:
\begin{enumerate}[label=$(\mathbb{LIP})$]
\item\label{lip}  
  The function $f^{k}$ belongs to $C^{1,2}([0,T_N] \times \erd)$ and is 
  globally Lipschitz, i.e. 
  $$|f^{k}(t, x) - f^{k} (t, y)| \leq K^{k} |x-y|, $$ 
  for every $t\in[0,T_N]$ and any $x,y \in \erd$, where $K^{k}>0$ is a 
  constant. 
\end{enumerate}
\begin{enumerate}[label=$(\mathbb{INT})$]
\item\label{int} 
  The process $X$ is an $\R^{d}$-valued semimartingale with absolutely 
  continuous characteristics $(b^N,c^N,F^N)$ under $\P_N$, such that the 
  following conditions hold
  \begin{equation}\label{eq:fin-mom}
  \int_{0}^{T_N} \int_{\R^{d}} \Big\{ |x|^{2} 1_{\{|x|\le1\}} + |x|    
    \e^{K |x|} 1_{\{|x|>1\}}  \Big\} F^N_{t}(\dx) \dt
    < C_{1}
  \end{equation}
  and
  \begin{equation}
  \int_{0}^{T_N} \|c^N_{t} \| \de t < C_{2},
  \end{equation}
  for some constants $C_{1}, C_{2}>0$ and $K=\sum_{k=1}^{N-1} K^{k}$.
\end{enumerate}
We denote by $\|\cdot\|$ the Euclidean norm on $\R^d$ and by 
$\langle\cdot,\cdot\rangle$ the associated scalar product.

\begin{remark}
The characteristic triplet of the semimartingale $X$ under the forward measure 
$\P_k$ is denoted by $(b^{k},c^{k},F^{k})$, while the truncation function can 
always be chosen the identity (i.e. $h(x)=x$) due to \eqref{eq:fin-mom}. 
Moreover, we use the standard conventions $\sum_\emptyset=0$ and 
$\prod_\emptyset=1$.
\end{remark}

\begin{theorem}\label{thm:back-ind}
Consider an $\R^{d}$-valued semimartingale $X$ and functions $f^{k}$ such that 
Assumptions \ref{lip} and \ref{int} are satisfied for each 
$k\in\bar{\mathcal{K}}$. Assume that the forward price processes are modeled 
via 
\begin{align}\label{eq:def-fpm-smmg}
F(t,T_{k},T_{k+1}) = \e^{f^k(t,X_t)}, \qquad t\in[0,T_k],
\end{align}
and the following \emph{drift condition} is satisfied
\begin{align}\label{eq:drift}
\la \De f^{k}(t, X_{t-}), b^N \ra \nonumber 
 &= - \fpoint^{k}(t, X_{t-}) - \frac{1}{2} \sum_{i,j=1}^{d} 
  \De_{ij}^{2} f^{k}(t, X_{t-}) (c_{t}^N)^{ij} \\
 &\quad - \frac{1}{2} \la \De f^{k}(t, X_{t-}), c_t^N 
    \De f^{k}(t, X_{t-})\ra   \nonumber \\ \tag{$\mathbb{DRIFT}$}
 &\quad - \sum_{j=k+1}^{N-1} \la \De f^{k}(t, X_{t-}), c_t^N 
    \De f^{j}(t, X_{t-})\ra \\ \nonumber    
 &\!\!\!\!\!\! - \int_{\R^{d}} \Big\{ \Big( 
    \e^{f^{k}(t,X_{t-}+x)-f^{k}(t,X_{t-})} - 1 \Big)
  \\ \nonumber
 & \times \prod_{j=k+1}^{N-1} 
    \e^{f^{j}(t,X_{t-}+x)-f^{j}(t,X_{t-})}
    - \la \De f^{k}(t, X_{t-}), x \ra \Big\} F^N_{t}(\dx)
\end{align}
for each $k \in \bar{\mathcal{K}}$. Then, the measures $(\P_{k})_{k \in 
\bar{\mathcal{K}}}$ defined via 
\begin{align}
\frac{\de \P_{k}}{\de \P_{k+1}} 
  = \frac{ \e^{f^{k}(T_k, X_{T_k})}}{ \e^{f^{k}(0, X_{0})}}
\end{align}
are equivalent forward measures and the forward prices processes 
$F(\cdot,T_{k},T_{k+1})$ are uniformly integrable martingales with respect to 
$\P_{k+1}$, for each $k \in \bar{\mathcal{K}}$. In particular, the model 
$(F(\cdot,T_k,T_{k+1}))_{k \in \bar{\mathcal{K}}}$ is arbitrage-free and 
satisfies Axioms \ref{a1} and \ref{a2}, as well as Property \ref{b1}.
\end{theorem}

\begin{proof}
The statement is proved via backward induction, motivated by the backward 
construction of \lib and forward price models.

\textit{First step:} 
We start from the forward price process $F(\cdot,T_{N-1},T_N)$ whose dynamics 
are 
\begin{equation*}
F(t,T_{N-1},T_N) = \e^{ f^{N-1}(t, X_t) }, \quad \quad t \in [0,T_{N-1}],
\end{equation*}
and examine its properties under the measure $\P_N$. The function $f^{N-1}$ 
satisfies \ref{lip} and the process $X$ satisfies \ref{int}, hence the process 
$f^{N-1}(\cdot,X)$ is an exponentially special 
\smmg by Proposition \ref{propUImart}. Using Proposition \ref{proplocmart}, we 
have that $F(\cdot,T_{N-1},T_N)$ is a $\P_N$-local martingale if the following 
condition holds:
\begin{align}\label{locmartcond}
&\!\!\!\! \la \De f^{N-1}(t, X_{t-}), b^N \ra  \nonumber\\
  &= -\fpoint^{N-1}(t, X_{t-}) -  \frac{1}{2} \sum_{i,j=1}^{d} \De_{ij}^{2} 
      f^{N-1}(t, X_{t-}) (c_{t}^N)^{ij} \nonumber \\
  &\quad - \frac{1}{2} \la \De f^{N-1}(t, X_{t-}), c^N 
    \De f^{N-1}(t, X_{t-})\ra  \\ \nonumber
  &\,\, - \int_{\erd} \Big( \e^{f^{N-1}(t,X_{t-}+x) - f^{N-1}(t,X_{t-})}
   - 1 - \la \De f^{N-1}(t, X_{t-}), x\ra\Big) F^N_{t}(\de x),
\end{align}
which is actually \eqref{eq:drift} for $k=N-1$. Moreover, Proposition 
\ref{propUImart} yields that $F(\cdot,T_{N-1},T_N)$ is even a $\P_N$-uniformly 
integrable martingale. Therefore, we can use $F(\cdot,T_{N-1},T_N)$ as a 
density process to define the measure $\P_{N-1}$ via
$$
\frac{\de \P_{N-1}}{\de \P_N} \Big|_{\lijepof_{\cdot}} 
  = \frac{ F(\cdot, T_{N-1}, T_N)}{ F(0, T_{N-1}, T_N)}
  = \frac{ \e^{f^{N-1}(\cdot, X)}}{ \e^{f^{N-1}(0, X_{0})}},
$$
and the characteristics of the process $X$ under the measure $\P_{N-1}$ are 
provided by
\begin{align*}
b^{N-1}_t &= b^N_t + c^N_t \De f^{N-1} (t,X_{t-}) \\
  &\quad + \int_{\erd} \left( \e^{f^{N-1}(t, X_{t-}+x)-f^{N-1}(t, X_{t-})} 
      - 1 \right) x F^N_t (\de x)\\
c^{N-1}_t &= c^N_t\\
F^{N-1}_t (\de x) &=  \e^{f^{N-1}(t, X_{t-}+x)-f^{N-1}(t, X_{t-})} 
   F^N_{t} (\de x);
\end{align*}
cf. Lemma \ref{SLgirsanov}. 

Then, we proceed backwards by considering the `next' forward price process 
$F(\cdot,T_{N-2},T_{N-1})$ with dynamics
\begin{align}\label{eq:F-mod-2}
F(t,T_{N-2},T_{N-1}) = \e^{ f^{N-2}(t,X_{t}) }, 
  \quad \quad t \in [0,T_{N-2}],
\end{align}
and verifying that subject to \ref{lip}, \ref{int} and \eqref{eq:drift} it is a 
$\P_{N-1}$-uniformly integrable martingale. Thus, it can be used as a density 
process to define the measure $\P_{N-2}$. 

Next, we provide the general step of the backward induction.

\textit{General step:} 
Let $k \in \{1,\ldots,N-1\}$ be fixed and consider the process $X$, the 
functions $f^{k+1},\ldots,f^{N-1}$ and the measures $\P_{k+1},\ldots,\P_N$ 
which are defined recursively via
$$
\frac{\de \P_{k+1}}{\de \P_{k+2}} \Big|_{\lijepof_{\cdot}} 
  = \frac{ F(\cdot, T_{k+1}, T_{k+2})}{ F(0, T_{k+1}, T_{k+2})}
  = \frac{ \e^{ f^{k+1}(\cdot,X)}}{ \e^{ f^{k+1}(0,X_{0})}}. 
$$ 
Assume that the forward price processes $F(\cdot,T_l,T_{l+1})$ have been modeled 
as exponential semimartingales according to \eqref{eq:def-fpm-smmg} and are 
$\P_{l+1}$-uniformly integrable martingales, for all $l\in\{k+1,\dots,N-1\}$. By 
repeatedly applying Lemma \ref{SLgirsanov}, we derive the 
$\P_{k+1}$-characteristics of $X$, which have the form
\begin{align}\label{kk-characteristics}
b_t^{k+1} &= b^N_{t} + c^N_{t} \sum_{j=k+1}^{N-1} 
  \De f^{j} (t, X_{t-}) \nonumber\\\nonumber
&\quad + \int_{\R^{d}} \left( \prod_{j=k+1}^{N-1} \e^{f^{j} (t, X_{t-}+x) 
  -   f^{j} (t, X_{t-})} - 1 \right) x F^N_{t}(\de x)\\
c_t^{k+1} &= c_{t}^N\\ \nonumber
F_t^{k+1} (\de x) &= \prod_{j=k+1}^{N-1} \e^{f^{j} (t, X_{t-}+x) 
  - f^{j} (t, X_{t-})} F_{t}^N (\de x).  
\end{align}

Now, the forward price process $F(\cdot,T_{k},T_{k+1})$ with dynamics
$$
F(t, T_{k}, T_{k+1}) = \e^{f^{k}(t, X_{t})}, \quad t \in [0, T_{k-1}],
$$
is a $\P_{k+1}$-local martingale if the following condition holds
\begin{eqnarray}\label{locmartcond_k}  
\nonumber  
\lefteqn{\la \De f^{k}(t, X_{t-}), b^{k+1} \ra}\\
\nonumber & = & - \fpoint^{k}(t, X_{t-})
- \frac{1}{2} \sum_{i,j=1}^{d} \De_{ij}^{2} f^{k}(t, X_{t-}) 
(c_{t}^{k+1})^{ij}\\
&& - \frac{1}{2} \la \De f^{k}(t, X_{t-}), c^{k+1} \De f^{k}(t, 
X_{t-})\ra \\
&& -  \nonumber \int_{\erd} \Big( \e^{f^{k}(t, 
X_{t-}+x)-f^{k}(t,X_{t-})}
	 - 1 -  \la \De f^{k}(t, X_{t-}), x\ra\Big) F^{k+1}_{t}(\dx);
\end{eqnarray}
cf. Proposition \ref{proplocmart}. By replacing \eqref{kk-characteristics} into 
\eqref{locmartcond_k} we see, after some straightforward calculations, that the 
latter is equivalent to the \eqref{eq:drift} condition. We can also verify that 
conditions \eqref{firstUI} and \eqref{secondUI} from  Proposition 
\ref{propUImart} hold for the function $f^{k}$ that satisfies \ref{lip} and 
the process $X$ that satisfies \ref{int}. Indeed, we have
\begin{eqnarray*}
\int_{0}^{T_N} \| c^{k+1}_{t}\| \de t 
 = \int_{0}^{T_N} \| c^N_{t}\| \de t 
 < C_{2},
\end{eqnarray*}
hence condition \eqref{secondUI} holds. Moreover, using \ref{lip} and \ref{int} 
we get that
\begin{eqnarray*}
\lefteqn{\int_{0}^{T_N} \int_{\R^{d}} (|x|^{2} \wedge 1) 
  F^{k+1}_{t}(\dx) \dt 
	+ \int_{0}^{T_N} \int_{|x|>1} |x| \e^{K^{k}|x|} 
	  F^{k+1}_{t}(\dx) \dt}\\
&=& \int_{0}^{T_N} \int_{\R^{d}} (|x|^{2} \wedge 1) \prod_{j=k+1}^{N-1}
	\e^{f^{j}(t, X_{t-}+x) - f^{j} (t, X_{t-})} F_{t}^N (\dx) \dt \\
&& + \int_{0}^{T_N} \int_{|x|>1} |x| \e^{K^{k-1} |x|} \prod_{j=k+1}^{N-1} 
	\e^{f^{j} (t, X_{t-}+x) - f^{j} (t, X_{t-})} F_{t}^N (\dx) \dt \\
&\leq & \int_{0}^{T_N} \int_{\erd} \Big\{ (|x|^{2} \wedge 1) 
	\e^{\sum_{j=k+1}^{N-1} K^{j} |x|}  + 1_{\{|x|>1\}} |x| 
	\e^{\sum_{j=k}^{N-1} K^{j} |x|} \Big\} F_{t}^N (\dx) \dt \\
&\leq & {\text{const}} \cdot \int_{0}^{T_N} \int_{\erd} \Big\{ |x|^{2} 
	1_{\{|x|\le1\}} + |x| \e^{K |x|} 1_{\{|x|>1\}} \Big\} F_{t}^N 
	(\dx)\dt \\
&<& C_{1},
\end{eqnarray*}
where the second to last inequality holds because the exponential function is 
bounded in the unit hypercube and the Lipschitz constants are positive. Hence, 
condition \eqref{firstUI} holds as well. Thus, Proposition \ref{propUImart} 
yields that the process $f^{k}(\cdot,X)$ is exponentially special and the 
forward price process $F(\cdot,T_{k},T_{k+1})$ is a $\P_{k+1}$-uniformly 
integrable martingale. 

Therefore, exactly as in the previous steps we can use the $\P_{k+1}$-uniformly 
integrable martingale $f^{k}(\cdot,X)$ to define the measure $\P_{k}$ via 
$$
\frac{\de \P_{k}}{\de \P_{k+1}} \Big|_{\lijepof_{\cdot}} 
	= \frac{ F(\cdot, T_{k}, T_{k+1})}{ F(0,T_{k},T_{k+1})}
	= \frac{\e^{f^{k}(\cdot,X)}}{\e^{ f^{k}(0, X_{0})}}. 
$$ 
Then, we can compute the $\P_{k}$-characteristics of $X$ using Lemma  
\ref{SLgirsanov} and consider the `next' forward price process with dynamics
$$
F(\cdot,T_{k-1},T_{k}) = \e^{f^{k-1}(\cdot,X^{k-1})}.
$$
This procedure produces an arbitrage-free semimartingale model for the forward 
price process, and thus also for the \lib rate, if the \eqref{eq:drift} 
condition holds for each $k \in \bar{\mathcal{K}}$.

Finally, we can easily show that the measures $\P_k$ are indeed forward 
measures, i.e. that $B(\cdot,T_l)/B(\cdot,T_k)$ is a $\P_k$-martingale for all 
$1\le k,l\le N$. This follows directly from Proposition III.3.8 in 
\citet{JacodShiryaev03}, using that
\begin{align*}
\frac{B(\cdot,T_l)}{B(\cdot,T_k)} \frac{\dd \P_k}{\dd \P_{l+1}}
 = \e^{f^l(\cdot,X)}
\end{align*}
which is a $\P_{l+1}$-martingale. 
\end{proof}

\begin{remark}\label{rem:BI-fm}
Let us point out that, although the true martingale property of the forward 
price process is not necessary to guarantee the absence of arbitrage, it is 
required in order to define the forward measures and to construct the model via 
backward induction. Aside from this, forward measures play a crucial role in 
term structure models since they allow to derive tractable formulas for 
interest rate derivatives. Indeed, the major advantage of forward measures for 
derivative pricing is that we can avoid the numerical computation of 
multidimensional integrals over joint distributions.
\end{remark}

\begin{remark}\label{r:xk-x-bc}
We may assume, if desired, that $d \geq N-1$ in order to ensure there are at 
least as many driving factors as the number of forward price processes. 
Moreover, by suitable choices of the functions $f^k$ we may select the 
components of $X$ driving a certain forward price process. See, for example, 
Section \ref{ss:LMM} where we work with an $N-1$-dimensional process $X$ and 
set $f^k(x)=\hat f^k(x_k)$, for $x=(x_1, \ldots, x_{N-1})$ and $\hat f^k: \er 
\to \er$, i.e. each forward price process is driven by a different component of 
the process $X$.
\end{remark}

\subsection{Modeling rates under the terminal measure}
\label{sec:term-meas}

Another possibility for constructing a model for the forward \lib rates, or 
equivalently the forward price processes, is to start with the family of forward 
price processes with respect to the terminal bond price $B(\cdot, T_N)$ in the 
tenor structure, i.e. 
$$
F(\cdot,T_{k},T_{N}) = \frac{B(\cdot, T_k)}{B(\cdot, T_N)}
$$
for all $k \in \bar{\mathcal{K}}$, and to model them simultaneously under the 
same measure, typically the terminal forward measure $\P_N$. Similarly to the 
previous section, the construction is based on specifying exponential 
semimartingale dynamics for the forward price process of the following 
functional form
\begin{align}
\label{terminal-forward-price}
F(\cdot,T_{k},T_{N}) = \e^{g^k(\cdot, X^k)},
\end{align}
where $g^k$ are suitable functions and $X^k$ are $d$-dimensional 
semimartingales, for $k \in \bar{\mathcal{K}}$.

Consider a collection of $\erd$-valued semimartingales $X^{k}=(X^{k}_{t})_{0\leq 
t \leq T_N}$ on $(\Omega,\lijepof,\mathbb{F},\P_N)$ and a collection of 
functions $g^{k}: [0,T_N] \times \erd \to \er$ for all $k \in 
\bar{\mathcal{K}}$, which satisfy the following assumptions:
\begin{enumerate}[label=$(\mathbb{LIP'})$]
\item\label{lip1}  
  The function $g^{k}$ belongs to $C^{1,2}([0,T_N] \times \erd)$ and is 
  globally Lipschitz, i.e. 
  $$|g^{k}(t, x) - g^{k} (t, y)| \leq \tilde K^{k} |x-y|, $$ 
  for every $t\geq 0$ and any $x,y \in \erd$, where $\tilde  K^{k}>0$ is a 
  constant.
\end{enumerate}
\begin{enumerate}[label=$(\mathbb{INT'})$]
\item\label{int1} 
  The process $X^k$ is an $\erd$-valued semimartingale with absolutely 
  continuous characteristics $(b^{k,N},c^{k,N},F^{k,N})$ under $\P_N$, such 
  that the following conditions hold
  \begin{equation}
  \int_{0}^{T_N} \int_{\erd} \Big\{ |x|^{2} 1_{\{|x|\le1\}} 
    + |x| \e^{\tilde K^{k}|x|} 
    1_{\{|x|>1\}} \Big\} F^{k,N}_{t}(\de x) \dt < \tilde C_{1}^{k}
  \end{equation}
  and
  \begin{equation}
  \int_{0}^{T_N} \|c^{k,N}_{t} \| \de t < \tilde C_{2}^{k},
  \end{equation}
  for some constants $\tilde C_{1}^{k}, \tilde C_{2}^{k}>0$. Recall that the 
  truncation can be chosen the identity.
\end{enumerate}

\begin{theorem}\label{th:terminal-measure-model}
Consider $\erd$-valued semimartingales $X^{k}$ and functions $g^{k}$ such that 
Assumptions \ref{lip1} and \ref{int1} are satisfied for each $k \in 
\bar{\mathcal{K}}$. Assume that the forward price processes are modeled via 
\begin{align}\label{eq:tmm}
F(t,T_{k},T_{N}) = \e^{g^k(t, X_t^{k})}, \quad t\in[0,T_k],
\end{align}
and the following \emph{drift condition} is satisfied
\begin{align}\label{terminal-locmartcond}
\nonumber & \!\!\!\!\!\!
\la \De g^{k}(t, X^k_{t-}), b_t^{k,N} \ra \\
  &= -\gpoint^{k}(t, X^k_{t-}) -  \frac{1}{2} \sum_{i,j=1}^{d} \De_{ij}^{2} 
      g^{k}(t, X^k_{t-}) (c_{t}^{k,N})^{ij} \nonumber \\
  &\quad	- \frac{1}{2} \la \De g^{k}(t, X^k_{t-}), c_t^{k,N} 
    \De g^{k}(t, X^k_{t-})\ra  \tag{$\mathbb{DRIFT'}$}\\ \nonumber
  &\quad - \int_{\erd} \Big( \e^{g^{k}(t,X^k_{t-}+x) - g^{k}(t,X^k_{t-})}
   - 1 - \la \De g^{k}(t, X^k_{t-}), x\ra\Big) F^{k,N}_{t}(\de x),
\end{align}
for all $k \in \bar{\mathcal{K}}$. Then, the forward price processes are 
uniformly integrable martingales with respect to the terminal forward measure 
$\P_{N}$, for all $k \in \bar{\mathcal{K}}$. In particular, the model is 
arbitrage-free and satisfies Axioms \ref{a1} and \ref{a2}, as well as Property 
\ref{b1}.
\end{theorem}

\begin{proof}
The proof is simpler compared to the proof of Theorem \ref{thm:back-ind} 
because we work only under the terminal measure $\P_N$. Furthermore, we can 
work simultaneously with all forward price processes $F(\cdot, T_k, T_N)$ for 
each $k \in \bar{\mathcal{K}}$. More precisely, for all $k \in 
\bar{\mathcal{K}}$ the function $g^{k}$ satisfies \ref{lip1} and the process 
$X^k$ satisfies \ref{int1}, hence the process $g^{k}(\cdot,X^k)$ is an 
exponentially special \smmg by Proposition \ref{propUImart}. Using Proposition 
\ref{proplocmart}, we get by virtue of the \eqref{terminal-locmartcond} 
condition that $F(\cdot,T_{k},T_N)$ is a $\P_N$-local martingale. Moreover, 
Proposition \ref{propUImart} yields that $F(\cdot,T_{k},T_N)$ is actually a 
$\P_N$-uniformly integrable martingale.
\end{proof}

\begin{remark}\label{r:xk-x}
In this construction we can use a family of semimartingales $X^k$, $k \in 
\bar{\mathcal{K}}$, where each forward price process is driven by a different 
semimartingale. This is possible because we do not have to perform measure 
changes as we did in the backward construction, since all forward price 
processes are modeled under a common measure. Hence, at this stage, we do not 
need to know the dependence structure between the processes $X^k$ which is 
necessary when applying Girsanov's theorem. However, for pricing purposes and 
also for linking the backward and the terminal measure constructions, we revert 
to a common $\erd$-valued driving process $X$ for which the dependence 
structure between its components is obviously fully known. Naturally, the 
dimension of the process $X$ can be chosen such that each rate is driven by a 
different component of the process; compare with Remark \ref{r:xk-x-bc}. 
\end{remark}

\begin{remark}
Based on \eqref{terminal-forward-price}, we can immediately deduce the dynamics 
of the forward price process $F(\cdot, T_k, T_{k+1})$ and the forward \lib rate 
$L(\cdot, T_k)$, for all $k \in \bar{\mathcal{K}}$. Using that
$$
1+\delta L(\cdot,T_k)
  = F(\cdot, T_k, T_{k+1}) 
  = \frac{F(\cdot, T_k, T_{N})}{F(\cdot, T_{k+1}, T_{N})},
$$
we obtain that 
\begin{align}\label{-forward-price}
1+\delta L(\cdot,T_k)
  = F(\cdot,T_{k},T_{k+1}) 
  = \e^{g^k(\cdot, X^k) - g^{k+1}(\cdot, X^{k+1})}.
\end{align}
\end{remark}

\begin{remark}
Assumptions \ref{lip1} and \ref{int1} are sufficient to produce an 
arbitrage-free family of \lib rates, but they are by no means necessary. 
Indeed, we can weaken them slightly by assuming that the functions $g^k$ 
satisfy \ref{lip1} and the processes $X^k$ have finite exponential moments. Then 
the previous theorem yields an arbitrage-free model that satisfies Axioms 
\ref{a1} and \ref{a2}, but not necessarily \ref{b1}. However, as pointed out 
also in Remark \ref{rem:BI-fm}, the latter is needed to define forward measures 
which are very useful because they typically lead to tractable pricing 
formulas. 
\end{remark}

\begin{remark}\label{r:connection}
Let us consider the case where all semimartigales $X^k$ coincide, i.e. $X^k 
\equiv X$ for all $k \in \bar{\mathcal{K}}$. Then, we can easily link the 
approach using backward induction presented in subsection \ref{sec:back-ind} 
and the approach under the terminal measure presented in this subsection. More 
precisely, starting from a family of functions $g^k$, $k \in 
\bar{\mathcal{K}}$, and a semimartingale $X$ satisfying \ref{lip1}, \ref{int1} 
and \eqref{terminal-locmartcond}, we define 
\begin{align}\label{eq:f-sp}
f^{k}(t, x):= g^{k}(t, x) - g^{k+1}(t, x).
\end{align}
The functions $f^k$ obviously satisfy \ref{lip} with the constants $K^k:= \tilde 
K^{k} + \tilde K^{k+1}$. Assume moreover that the semimartingale $X$ satisfies 
\ref{int} with $K^k$ as above. Then the model for the terminal forward prices 
given by \eqref{terminal-forward-price} can be equivalently written as 
$$
F(\cdot,T_{k},T_{k+1}) 
  = \e^{f^k(\cdot, X)},  \qquad k \in \bar{\mathcal{K}},
$$
with $f^k$ given by \eqref{eq:f-sp} and all assertions of Theorem 
\ref{thm:back-ind} remain valid. 

Conversely, assuming that a model for the forward prices \eqref{eq:def-fpm-smmg} 
is given via a family of functions $f^k$, $k \in \bar{\mathcal{K}}$, and a 
semimartingale $X$ satisfying \ref{lip}, \ref{int} and \eqref{eq:drift}, we 
define 
\begin{align}\label{eq:g-sp}
g^{k}(t, x):= \sum_{j=k}^{N-1} f^{j}(t, x).
\end{align}
The functions $g^k$ satisfy condition \ref{lip} with the constants $\tilde 
K^k:= \sum_{j=k}^{N-1} K^{j}$. Assuming furthermore that the semimartingale $X$ 
satisfies \ref{int1} with $\tilde K^k$ as above, the model for the forward 
prices \eqref{eq:def-fpm-smmg} can be equivalently written as 
$$
F(\cdot,T_{k},T_{N}) 
  = \e^{g^k(\cdot, X)}, \qquad k \in \bar{\mathcal{K}},
$$
with $g^k$ defined in \eqref{eq:g-sp}. This easily follows from the following 
telescopic product 
\begin{align}\label{eq:telescopic}
F(\cdot,T_{k},T_{N}) 
  = \frac{B(\cdot, T_k)}{B(\cdot, T_N)} 
  = \prod_{j=k}^{N-1} \frac{B(\cdot, T_{j})}{B(\cdot, T_{j+1})} 
  = \prod_{j=k}^{N-1} F(\cdot,T_{j},T_{j+1}).
\end{align}
Thus, we conclude that Theorem \ref{th:terminal-measure-model} is valid for the 
semimartingale $X$ and the functions $g^k$, $k \in \bar{\mathcal{K}}$.
\end{remark}

\subsection{Observations and ramifications}

Next, we discuss further properties of the models constructed in the previous 
two subsections. In particular, we derive conditions such that a \lib model is 
structure preserving and produces non-negative rates. In order to provide a 
unified treatment of both modeling approaches, we assume that $X^k\equiv X$ in 
subsection \ref{sec:term-meas}, for all $k \in \bar{\mathcal K}$.

\begin{lemma}
\emph{(i)} If the functions $f^k$  are non-negative for all $k \in 
\bar{\mathcal{K}}$, then the \lib rates in the model \eqref{eq:def-fpm-smmg} are 
non-negative, i.e. Property \ref{c} is satisfied. \par
\emph{(ii)} If the functions $g^k$ are non-negative and such that $g^k(t, x)  
\geq g^{k+1}(t, x)$ for all $k \in \bar{\mathcal{K}}$ and all $(t, x) \in  [0, 
T_N] \times \erd$, then the \lib rates in the model \eqref{eq:tmm} are 
non-negative, i.e. Property \ref{c} is satisfied. 
\end{lemma}

\begin{proof}
This follows directly from the relation between forward prices and \lib rates, 
see \eqref{eq:connection-Libor-forward-price} and \eqref{-forward-price}.
\end{proof}

The second tractability property \ref{b2} states that a LIBOR model is 
\textit{structure preserving} if the characteristics of the driving process are 
transformed in a deterministic way under different forward measures, which 
ensures that the driving processes remain in the same class under all forward 
measures. In order to formalize the statement, we consider the following 
assumption. 

\begin{enumerate}[label=$(\mathbb{E})$]
\item\label{e}
Let $U:=\erd$ (respectively $U:=\er_+^d$). The measure 
$\P_N^{X_{t-}}$ is absolutely continuous with positive Lebesgue density on $U$, 
for all $t\in[0,T_N]$.
\end{enumerate}

\noindent We say that a LIBOR model is structure preserving if the tuple 
$(\beta^l,Y^l)$ defining the change of measure from the forward measure $\P_l$ 
to $\P_{l-1}$, for $l=N,\ldots, 1$, via Girsanov's theorem as in Lemma 
\ref{SLgirsanov}, is deterministic.

Notice that under assumption \ref{e}, $\beta^l$ is deterministic if and only if 
$Y^l$ is so; indeed, if $Y^{l}(t,x) = \e^{f^{l} (t, X_{t-}+x) - f^{l} (t, 
X_{t-})} $ is assumed to be deterministic, the function $f^l$ must satisfy
\begin{align}
f^l(t,y+x) - f^l(t,y) = h^l(t,x)
\end{align}
for every $x,y\in U$, for some function $h^l$. Taking derivatives with respect 
to $y$, we get that
\[ \De f^l(t,y+x)= \De f^l(t,y) \]
for every $x,y\in U$, hence $\De f^l(y)$ is constant, and thus 
$f^l(t,\cdot)|_U$ is an affine function. This implies that $\beta^l_t=\De 
f^l(t,X_{t-})$ is deterministic.

Conversely, assume that the variable $\beta^l_t=\De f^{l} (t, X_{t-})$ is 
deterministic. Since the support of $\P_N^{X_{t-}}$ is $U$, $\De f^{l}$ is 
continuous and $\P_N^{X_{t-}}$ has a positive Lebesgue measure on $U$, we 
obtain that $\De f^{l}$ is constant and hence $f^{l}$ is affine in the 
second variable. Thus, we conclude that $Y^l(t,x)$ is deterministic.

The next result provides necessary and sufficient conditions for \ref{b2} to be 
satisfied.

\begin{proposition} \label{p:structure-preserve}
If the functions $f^k$ and $g^k$ are affine in the second variable for every $k 
\in \bar{\mathcal{K}}$, then the \lib models in \eqref{eq:def-fpm-smmg} and 
\eqref{eq:tmm} are structure preserving. Conversely, assume \ref{e}. If the \lib 
models in \eqref{eq:def-fpm-smmg} and \eqref{eq:tmm} are structure preserving, 
then the functions $f^k|_{[0,T_N]\times U}$ and $g^k|_{[0,T_N]\times U}$ are 
affine in the second variable for every $k \in \bar{\mathcal{K}}$.
\end{proposition}

\begin{proof}
We will concentrate on the model constructed by backward induction, while the 
other one follows analogously. Following the argumentation in the proof of 
Theorem \ref{thm:back-ind}, the characteristics of $X$ under $\P_l$ have the 
following form
\begin{align}\label{kl-characteristics}
b_t^{l} &= b^{N}_{t} + c^{N}_{t} \sum_{j=l+1}^{N-1} 
  \De f^{j} (t, X_{t-}) \nonumber\\\nonumber
&\quad + \int_{\erd} \left( \prod_{j=l+1}^{N-1} \e^{f^{j} (t, X_{t-}+x) - 
  f^{j} (t, X_{t-})} - 1 \right) x F^{N}_{t}(\de x)\\
c_t^{l} &= c_{t}^{N}\\ \nonumber
F_t^{l} (\de x) &= \prod_{j=l+1}^{N-1} \e^{f^{j} (t, X_{t-}+x) 
  - f^{j} (t, X_{t-})} F_{t}^{N} (\de x).  
\end{align}
Assume that the function $f^k(t,x)$ is affine in $x$, i.e. there exist 
$\alpha^k(t)$ and $\beta^k(t)$ such that $f^k(t,x)=\alpha^k(t) + \la 
\beta^k(t),x\ra$, then we can easily deduce that $(b^{l},c^{l},F^{l})$ in 
\eqref{kl-characteristics} is only a deterministic transformation of 
$(b^{N},c^{N},F^{N})$. 

The converse statement is already implied by the arguments preceding this 
Proposition. 
\end{proof}

The statement of Proposition \ref{p:structure-preserve} can be generalized to 
allow for more general driving processes. Assumption \ref{e}, for instance, can 
be formulated for more general sets $U$. As an example, $X$ could be a process 
that is positive in some coordinate and real or negative in another. On the 
other hand, processes with fixed jump sizes, such as the Poisson process, 
require a slightly different approach than in the proof above, taking care of 
the state space of the process and the support of the jump measure.

The following remark summarizes further interesting properties of \lib rates 
that can be easily deduced from this general modeling framework.

\begin{remark}\label{rem-affine-boudedfrombelow}
If the function $f^k$ is affine in the second argument, i.e.
\begin{equation}\label{affinefk}
f^k(t,x) = \alpha^k(t) + \la\beta^k(t),x\ra \,,
\end{equation}
with functions $\alpha^k,\beta^k\in C^1(\er_+)$ and the process is required to 
satisfy
\[F(\cdot,T_k,T_{k+1})\ge 1,\]
i.e.\ produce non-negative \lib rates, then the process $X$ has to be bounded 
from below.
\end{remark}
\section{Examples} 
\label{sec:4}

\subsection{\lib market models}
\label{ss:LMM}

We start by revisiting the class of \lib market models in view of the general 
framework developed in the previous section. We will concentrate on the L\'evy 
\lib model of \citet{EberleinOezkan05} in order to fix ideas and processes, and 
as a representative of other LIBOR market models which fit in this framework as 
well, such as models with local volatility, stochastic volatility or driven by 
jump-diffusions. See, among many other references, \citet{BrigoMercurio06}, 
\citet{Schoenmakers05}, \citet{Glasserman03} and 
\citet{Andersen_Piterbarg_2010}.

We assume that the driving process $X$ is an $\er^{N-1}$-valued semimartingale 
of the form 
\begin{align}
\label{eq:levy}
X & = B + \Lambda \cdot L,
\end{align}
where $L$ is an $\er^n$-valued time-inhomogeneous L\'evy process with 
characteristic triplet $(0, c^L, F^L)$ under the terminal measure $\P_{N}$ with 
respect to the truncation function $h(x)=x$, and $\Lambda = [\lambda(\cdot,T_1), 
\ldots, \lambda(\cdot,T_{N-1}) ]$ is an $(N-1)\times n$ volatility matrix where, 
for every $k\in\bar{\mathcal{K}}$, $\lambda(\cdot, T_k)$ is a deterministic, 
$n$-dimensional function. Moreover, $\Lambda \cdot L$ denotes the It\^{o} 
stochastic integral of $\Lambda$ with respect to $L$, while the drift term $B = 
\int_0^\cdot b(s)\ds = \left(\int_0^\cdot b(s,T_1) \ds,\ldots,\int_0^\cdot 
b(s,T_{N-1})\ds\right)$ is an $(N-1)$-dimensional stochastic process. We further 
assume that the following exponential moment condition is satisfied: 
\begin{enumerate}[label=$(\mathbb{EM})$] 
\item\label{EM-fpm}
  Let $\varepsilon>0$ and $M>0$, then
  \begin{equation*}
  \int_0^{T_N}\int_{|x|>1}\e^{\la u, x \ra} F^L_s(\dx)\ds <\infty 
  \qquad\text{for all $ u\in[-(1+\varepsilon) M, (1+\varepsilon) M]^n$};
  \end{equation*}
\end{enumerate}
while the volatility functions satisfy:
\begin{enumerate}[label=$(\mathbb{VOL})$] 
\item\label{vol} 
  The volatility $\lambda(\cdot,T_k):\,[0,T_N]\to\er^n_+$  is a deterministic, 
  bounded function such that for  $s>T_k$, $\lambda(s,T_k) =0$, for every $k 
  \in \bar{\mathcal{K}}$. Moreover, 
  \begin{equation}\label{eq:fpm.volbound}
  \sum_{k=1}^{N}  \lambda^j(s,T_k) \leq M \qquad\qquad\text{for all }
  s\in[0,T_N],
  \end{equation}
  for every   $s \in [0, T_N]$ and every coordinate $j \in \{1, \ldots, n\}$.
\end{enumerate}

The construction of the \lev \lib model will follow the backward induction 
approach of subsection \ref{sec:back-ind}. Define, for all 
$k\in\bar{\mathcal{K}}$ and $x=(x_1, \ldots, x_{N-1}) \in \er^{N-1}$, the 
functions 
\begin{align}
f^k(t,x) &:= \log\big(1+\delta_k L(0,T_k)\e^{x_k}\big)
\end{align}
and set
\begin{align}
F(t, T_k, T_{k+1}) &= \e^{f^k(t,X_t)}, \qquad k\in\bar{\mathcal{K}}.
\end{align}
Then, it follows easily that
\begin{align*}
F(t, T_k, T_{k+1}) 
& = 1+\delta_k L(0,T_k)\e^{X^k_t},
\end{align*}
which coincides with the dynamics of the L\'evy \lib model of Eberlein and \"Ozkan
\citeyearpar{EberleinOezkan05}, that are provided by 
\begin{align}
L(t, T_k) & = L(0,T_k) \exp \left( \int_0^t b(s, T_k) \ds + \int_0^t \lambda(s, 
T_k) \dd L_s   \right).
\end{align}

The function $f^k$ is Lipschitz continuous with constant~$1$, hence condition 
\ref{lip} is satisfied with $K^k=1$ for each $k\in\bar{\mathcal K}$. Moreover, 
thanks to assumptions \ref{EM-fpm} and \ref{vol}, condition \ref{int} is also 
satisfied for every $k \in \bar{\mathcal K}$. Therefore, an application of 
Theorem \ref{thm:back-ind} yields the drift $B^k=\int_0^\cdot b(s,T_k)\ds$ of 
this model under the terminal measure. More precisely, we have that
$$
\partial_{x_k}  f^k(t,x)
 = \frac{\delta_k L(0,T_k) \e^{x_k}}{1+\delta_k L(0,T_k) \e^{x_k}} 
 =: \ell^k(x_k),
$$
and $\partial_{x_j}  f^k(t,x) = 0$, for $j \neq k$, while also $\frac{\mathrm 
d}{\dt}f^k(t,x)=0$. Moreover, 
$$
\partial_{x_k x_k}  f^k(t,x) 
  = \frac{\delta_k L(0,T_k)\e^{x_k}}{(1+\delta_k L(0,T_k)\e^{x_k})^2},
$$
and $\partial_{x_i x_j}  f^k(t,x) = 0$, for all $(i, j) \neq (k, k)$. According 
to Proposition 2.4 in \citet{Kallsen06},  the $\P_N$-characteri\-stics $(b^N, 
c^N, F^N)$ of $X$ are given by
\begin{align}
\label{eq:char-X-LMM}
\notag b_t^N & = b(t) \\
c_t^N & = \big\langle\Lambda(t), c^L_t \Lambda(t)\big\rangle \\
\notag F_t^N(A) & = \int_{\er^n} \indik_{A} (\Lambda(t) x) F^L_t(\dx), \qquad A 
\in \mathcal{B}(\er^{N-1}) \setminus \{ 0\},
\end{align}
hence the \eqref{eq:drift} condition from Theorem \ref{thm:back-ind} becomes 
\begin{align}
\label{eq:LMM-drift}
\ell^k(X_{t-}^k) b(t, T_k) \notag & = -\frac{1}{2} \frac{\ell^k(X^k_{t-})}{1+\delta_k 
L(0,T_k)\e^{X^k_{t-}}} \langle \lambda(t, T_k), c^L_t \lambda(t, T_k)\rangle \\
\notag &   -\frac{1}{2} (\ell^k(X^k_{t-}))^2 \langle \lambda(t, T_k), c^L_t 
\lambda(t, T_k)\rangle \\
 &  - \sum_{j=k+1}^{N-1} \ell^k(X^k_{t-}) \ell^j(X^j_{t-}) \langle 
\lambda(t, T_k), c^L_t \lambda(t, T_j)\rangle \\
\notag & - \int_{\er^{N-1}} \Big[ \left(\e^{f^k(X^k_{t-}+x_k) -  f^k(X^k_{t-})} 
-1 \right) \prod_{j=k+1}^{N-1}\left(\e^{f^j(X^j_{t-}+x_j) -  f^j(X^k_{t-})} 
\right) \\
\notag & \qquad \qquad - \ell^k(X^k_{t-} x_k) \Big] F_t^X(\dx)
\end{align}
Notice that
\begin{align*}
\frac{\ell^k(X^k_{t-})}{1+\delta_k L(0,T_k)\e^{X^k_{t-}}}  + 
(\ell^k(X^k_{t-}))^2 & = \frac{\delta_k L(0,T_k)\e^{X^k_{t-}} + (\delta_k 
L(0,T_k)\e^{X^k_{t-}})^2}{(1+\delta_k L(0,T_k)\e^{X^k_{t-}})^2} \\
& = \frac{\delta_k L(0,T_k)\e^{X^k_{t-}} }{1+\delta_k L(0,T_k)\e^{X^k_{t-}}} = 
\ell(X^k_{t-})
\end{align*}
and that, for all $j=k, \ldots, N-1$,
\begin{align}
\notag \e^{f^j(X^j_{t-}+x_j) -  f^j(X^k_{t-})} 
& = \frac{1+\delta_j L(0,T_j)\e^{X^j_{t-} +x_j}}{1+\delta_j 
L(0,T_j)\e^{X^j_{t-}}} \\
\notag & = \frac{1+\delta_j L(0,T_j)\e^{X^j_{t-}} + \delta_j 
L(0,T_j)\e^{X^j_{t-}} (\e^{x_j}-1)}{1+\delta_j L(0,T_j)\e^{X^j_{t-}}} \\
\notag & = 1 + \ell^j(X^j_{t-}) (\e^{x_j} -1).
\end{align}
Inserting the above simplifications into \eqref{eq:LMM-drift} yields
\begin{align}
\label{eq:random-terms-LMM}
b(t, T_k) & = -\frac{1}{2} \langle \lambda(t, T_k), c^L_t \lambda(t, 
T_k)\rangle 
-  \sum_{j=k+1}^{N-1}\ell^j(X^j_{t-}) \langle \lambda(t, T_k), c^L_t \lambda(t, 
T_j)\rangle \notag \\ \notag
&\quad  - \int_{\er^N} \Big[  (\e^{x_k} -1)  \prod_{j=k+1}^{N-1} \left(1 + 
\ell^j(X^j_{t-}) (\e^{x_j} -1) \right)
 - x_k \Big] F_t^X(\dx) \\ \notag 
 & = -\frac{1}{2} \langle \lambda(t, T_k), c^L_t \lambda(t, T_k)\rangle -  
\sum_{j=k+1}^{N-1}\ell^j(X^j_{t-}) \langle \lambda(t, T_k), c^L_t \lambda(t, 
T_j)\rangle \\
 &\quad  - \int_{\er^n} \Big[  (\e^{\langle \lambda(t, T_k), y \rangle } -1)  
\prod_{j=k+1}^{N-1} \left(1 + \ell^j(X^j_{t-}) (\e^{\langle \lambda(t, T_j), y 
\rangle} -1) \right) \\\notag
 & \qquad \qquad 
 - \langle \lambda(t, T_k), y \rangle \Big] F^L_t(\dy),
\end{align}
where the second equality follows by \eqref{eq:char-X-LMM}. The equation above 
now can be recognized as the drift condition of the \lev \lib model; cf. 
Papapantoleon, Schoenmakers, and Skovmand\citeyearpar[eq.~(2.7)]{PapapantoleonSchoenmakersSkovmand10}.

\begin{remark}\label{rem-LevyLibor-satisfiesproperties}
The \lib market models satisfy Axioms \ref{a1} and \ref{a2}, as well as 
Properties \ref{b1} and \ref{b4} by construction. On the other hand, Properties 
\ref{b2} and \ref{b3} are not satisfied. Regarding \ref{b2}, this follows 
immediately by Proposition \ref{p:structure-preserve} (at least for driving 
processes satisfying \ref{e}, which is typically the case), since the 
functions $f^k$, $k \in \bar{\mathcal K}$, are not affine in the second 
argument. Moreover, the drift term \eqref{eq:random-terms-LMM} which contains 
the random terms $\delta_jL(t,T_j)/(1+\delta_jL(t,T_j))$ implies that the vector 
of \lib rates $(L(\cdot, T_k))_{k \in \bar{\mathcal K}}$ considered as a whole 
is Markovian, but not the single \lib rates, because their dynamics depend on 
the other rates as well. Hence, \ref{b3}  does not hold. Finally, Property 
\ref{c} is obviously satisfied in this model.
\end{remark}

\subsection{L\'evy forward price models} 
\label{ss:FPM}

Next, we show that the L\'evy forward price models can be easily embedded in our 
general framework starting from the terminal measure construction; starting from 
the backward induction approach is even easier. The L\'evy forward price models 
were introduced by \citet[pp.~342-343]{EberleinOezkan05}; see also 
\citet{Kluge05} for a detailed construction and \citet{KlugePapapantoleon06} for 
a concise presentation.

We will model the dynamics of the forward price relative to the terminal bond 
price under the terminal measure $\P_N$, via
\begin{align}\label{eq:fpm-tm}
F(t,T_k,T_N) = \e^{g^k(t,X^k_t)},
\end{align}
where the function $g^k$ is of the following affine form
\begin{align}\label{eq:fpm-g}
g^k(t,x) := \log F(0,T_k,T_N) + x,
\end{align}
while the process $X^k$ has the following dynamics
\begin{align}\label{eq:fpm-x}
X^k := \int_0^\cdot b_s^{k,N} \ds 
      + \sum_{i=k}^{N-1} \int_0^\cdot \lambda(s,T_i) \dd L_s. 
\end{align}
The driving process $L$ and the volatility functions $\lambda(\cdot,T_i)$ are 
specified, while the drift term $b^{k,N}$ is determined by the no-arbitrage 
\eqref{terminal-locmartcond} condition. In particular, $L$ is an $\er^n$-valued 
time-inhomogeneous \lev process with $\P_N$-local characteristics $(0,c^L,F^L)$ 
satisfying condition \ref{EM-fpm} and the volatility functions satisfy condition 
\ref{vol}. The function $g^k$ trivially satisfies the \ref{lip1} condition with 
constant 1, while the process $X^k$ satisfies the \ref{int1} condition by virtue 
of \ref{EM-fpm} and \ref{vol}; see also 
\citet[Remark~3.7]{Criens_Glau_Grbac_2015}. Therefore, we can apply Theorem 
\ref{th:terminal-measure-model} and, after some computations, the 
\eqref{terminal-locmartcond} condition yields that
\begin{align}\label{eq:fpm-b}
b_t^{k,N} = -\frac12  c_t^{k,N} - \int_\er (\e^x-1-x) F_t^{k,N}(\dx).
\end{align}
Moreover, using \citet[Lemma~3]{KallsenShiryaev02b}, the $\P_N$-local 
characteristics of the stochastic integral process $X^k$ are 
\begin{align}\label{eq:fpm-c}
c_t^{k,N} 
  = \left\la \sum_{i=k}^{N-1} \lambda(t,T_i),
    c^L_t \sum_{i=k}^{N-1} \lambda(t,T_i) \right\ra
\end{align}
and 
\begin{align}\label{eq:fpm-f}
F_t^{k,N}(A) 
  = \int_{\er^n}\indik_A\left(  \sum_{i=k}^{N-1}  \left\la\lambda(t,T_i),  x  
  \right\ra \right) F^L_t(\dx), \quad A\in\mathcal{B}(\er).
\end{align}

Now, using \eqref{eq:fpm-tm}--\eqref{eq:fpm-x}, we get that the dynamics of the 
forward price process $F(\cdot,T_k,T_{k+1})$ are provided by
\begin{align*}
F(t,T_k,T_{k+1})
  &= \frac{F(t,T_k,T_N)}{F(t,T_{k+1},T_N)}
   = F(0,T_k,T_{k+1}) \e^{X^k_t - X^{k+1}_t} \nonumber \\
  &= F(0,T_k,T_{k+1}) \exp\left( \int_0^t \big(b_s^{k,N}-b_s^{k+1,N}\big) \ds 
  	+ \int_0^t \lambda(s,T_k) \dd L_s \right),
\end{align*}
hence the forward price process is driven by its corresponding volatility 
function and the time-inhomogeneous \lev process, as specified in the \lev 
forward process models. We just have to check that the drift terms coincide as 
well. Indeed, using  \eqref{eq:fpm-b}--\eqref{eq:fpm-f}, after some 
straightforward calculations we get that
\begin{align*}
b_s^{k,N}-b_s^{k+1,N}
 &= - \frac12 \left\la \lambda(s,T_k), c^L_s \lambda(s,T_k) \right\ra 
  - \sum_{i=k+1}^{N-1} \left\la \lambda(s,T_k), c^L_s \lambda(s,T_i) \right\ra  
\\
 &\, - \int_{\er^n} \left\{ \Big( \e^{\la \lambda(s,T_k), x\ra }-1 \Big) 
 \e^{ \sum_{i=k+1}^{N-1} \la\lambda(s,T_i), x\ra  } - \la \lambda(s,T_k), x\ra 
  \right\} F^L_s(\dx),
\end{align*}
which is exactly the $\P_N$-drift of the forward price process; compare with Kluge and Papapantoleon
\citeyearpar[eqs.~(19)--(21)]{KlugePapapantoleon06}.

\begin{remark}
The \lev forward price model satisfies Axioms \ref{a1} and \ref{a2} as well as 
Properties \ref{b1}  and \ref{b4} by construction. Moreover, it satisfies 
Properties  \ref{b2} and \ref{b3}; cf. Proposition \ref{p:structure-preserve}. 
Property \ref{c} is not satisfied however, i.e. the \lib rates can become 
negative; cf. Remark \ref{rem-affine-boudedfrombelow}. 
\end{remark}

\subsection{Affine \lib models} 
\label{sec:alm}

Finally, we examine a class of \lib models driven by affine processes, and in 
particular the affine \lib models proposed by  
\citet{KellerResselPapapantoleonTeichmann09}. Our main reference for the 
definition and properties of affine processes is 
\citet{DuffieFilipovicSchachermayer03}. 

Let $X=(X_t)_{0\leq t \leq T_N}$ be a conservative affine process according to 
Definitions 2.1 and 2.5 in \citet{DuffieFilipovicSchachermayer03} with state 
space $D=\er_+$. We consider a one-dimensional process here only for notational 
simplicity; the $d$-dimensional case can be treated in exactly the same manner. 
Moreover, the state space is restricted to the positive half-line following 
\citet{KellerResselPapapantoleonTeichmann09}, which is necessary in order to 
produce a model satisfying \ref{c}; see also Remark 
\ref{rem-affine-boudedfrombelow}. We can equally well choose the state space 
$D=\er$, and then interest rates in the model will also take negative values. 

The process $X$ is a semimartingale with absolutely continuous characteristics, 
and the local characteristics $(b^X,c^X,F^X)$ of $X$ with respect to the 
truncation function $h(x):=1\wedge x$, for $x \in D$, are given as
\begin{eqnarray*}
b^X_t &=& \tilde{b} + \beta X_{t-} \\
c^X_t &=& 2  \alpha X_{t-}  \\
F^X_t(\de \xi) &=& F^1(\de \xi) + X_{t_-}F^2(\de \xi)
\end{eqnarray*}
for some $\tilde{b}>0$, $ \beta \in\er$, $\alpha>0$ and L\'evy measures $F^1$ 
and $F^2$ on $D \setminus \{0\}$ (cf. Theorem 2.12 in 
\citet{DuffieFilipovicSchachermayer03}), with 
$$
\tilde b:= b + \int_{\xi>0} h(\xi) F^1(\de \xi)\,.
$$
Affine processes are characterized by the following property of their moment 
generating function:
\begin{align}\label{affine-def}
\E_\mathrm{x}\big[\exp (u X_t )\big]
 = \exp\big( \phi(t,u) +  \psi(t,u) \mathrm{x}   \big),
\end{align}
for all $(t,u,\mathrm{x}) \in [0,T_N] \times \I_T \times D$, where 
$\E_\mathrm{x}$ denotes the expectation with respect to $\P_\mathrm{x}$---a 
probability measure such that $X_0=\mathrm{x} \in D$, $\P_\mathrm{x}$-a.s. 
Moreover, the set $ \I_T $ is defined by 
 \begin{align}
\label{eq:I_T}
\mathcal{I}_T
:= \big\{ u\in\er: \E_\mathrm{x}\big[\e^{u X_{T_N}}\big] < \infty,
        \,\,\text{for all}\; \textrm{x} \in D \big\},
\end{align}
while $(\phi,\psi)$ is a pair of deterministic functions $\phi, \psi :[0, T_N] 
\times\I_T\to\er$. The functions $\phi$ and  $\psi$ are given as solutions to 
generalized Riccati equations (cf. Theorem 2.7 in  
\citet{DuffieFilipovicSchachermayer03}), that is
\begin{align}\label{eq:Riccati1}
\begin{split}
\partial_t  \phi(t,u) &= F \big(\psi(t,u) \big)\,,\quad \phi(0,u) = 0 \\
\partial_t \psi(t,u) &= R \big( \psi(t,u) \big)\,,\quad \psi(0,u) = u,
\end{split}
\end{align}
where
\begin{align}\label{eq:Riccati2}
\begin{split}
F(u) &= bu + \int_{\xi>0} \big(\e^{u\xi} - 1 \big) F^1(\de \xi), 
\\
R(u) &= \alpha u^2+ \beta u + \int_{\xi>0} \big( \e^{u\xi} - 1 - 
	uh(\xi) \big) F^2(\de \xi) \,.
\end{split}
\end{align}

We introduce next the class of \textit{affine forward price models}, where the 
forward price is an exponentially-affine function of the driving affine process 
$X$. In particular, we consider the setting of the terminal measure construction 
of subsection \ref{sec:term-meas} with 
\begin{align}\label{eq:afpm}
g^k(t,x) = \theta^k(t)  + \vartheta^k(t)x
\quad\text{ and }\quad
X^k \equiv X.
\end{align}
The next result shows that the functions $\theta^k, \vartheta^k$ are solutions 
to generalized Riccati equations themselves.

\begin{proposition}\label{prop:affine}
Let $X$ be an affine process with values in $D$, $X_0=1$ and satisfying 
\ref{int1}, and $g^k$, $k \in \bar{\mathcal K}$, be a collection of functions 
given by \eqref{eq:afpm} where $\theta^k, \vartheta^k: [0, T_N] \to \er$ are 
deterministic functions of class $C^1$. Then, the forward price process given 
by 
\begin{align}
\label{eq:affine-model-2}
F(t, T_k, T_N) = \e^{ \theta^k(t) + \vartheta^k(t) X_t },
\quad t\in[0,T_k],
\end{align}
is a uniformly integrable martingale, for all $k \in \bar{\mathcal K}$, if the 
functions $\theta^k$ and $\vartheta^k$ satisfy 
\begin{align}
\begin{split}
\partial_t \theta^k(t)    &= - F \big( \vartheta^k(t) \big), \\
\partial_t \vartheta^k(t) &= - R\big( \vartheta^k(t) \big),
\end{split}
\end{align}
with $F$ and $R$ given by \eqref{eq:Riccati2}.
\end{proposition}

\begin{proof}
The process $X$ satisfies \ref{int1} by assumption, while the functions $g^k$ 
satisfy \ref{lip1}. Therefore, we can apply Theorem 
\ref{th:terminal-measure-model} and the result follows after straightforward 
calculations, by inserting the characteristics of $X$ into the 
\eqref{terminal-locmartcond} condition and using that
\begin{equation*}
\partial_t g^k(t,x) =  \partial_t \theta^k(t) + \partial_t \vartheta^k(t) x, \ \
\partial_x  g^k(t, x) = \vartheta^k(t), \ \ 
\partial_{xx}  g^k(t, x) = 0. \qedhere
\end{equation*}
\end{proof}

\begin{remark}
The affine forward price models given by \eqref{eq:affine-model-2} satisfy 
Axioms \ref{a1} and \ref{a2} as well as Properties \ref{b1}--\ref{b3}. Property 
\ref{b4} is not satisfied and the model has to be calibrated to the initial 
term structure, similarly to short rate models. Indeed, notice that the initial 
forward price $F(0,T_k,T_N)$ does not appear in the function $g^k(t,x)$, 
contrary to the previous two examples.\footnote{We could, of course, use the 
following affine function $g^k(t,x) = \log F(0,T_k,T_N) + \theta^k(t) + 
\vartheta^k(t)x$ and the model fits automatically the initial term structure. 
However, it becomes then difficult to provide models that produce non-negative 
LIBOR rates.} Moreover, these models satisfy Property \ref{c} if and only if 
the functions $\theta^k$ and $\vartheta^k$ are non-negative; compare also with 
Remark \ref{rem-affine-boudedfrombelow}.
\end{remark}

The \alms introduced by \citet{KellerResselPapapantoleonTeichmann09} can  
naturally be embedded in this construction. More precisely, we have the 
following.

\begin{corollary}
The \alms whose dynamics are provided by 
$$
F(t, T_k, T_N) 
  = \E_N \left[  \e^{u_k X_{T_N}} | \mathcal F_t \right] 
  = \e^{\phi(T_N-t, u_k) + \psi(T_N-t,  u_k) X_t } ,
$$
with parameters $u_k \in \er_+$ for $k \in \bar{\mathcal K}$, is a special case 
of the affine forward price models with
\begin{align*}
\theta^k(t) := \phi(T_N-t, u_k)
\quad \text{ and } \quad 
\vartheta^k(t) :=  \psi(T_N-t, u_k),
\end{align*}
where $\phi(\cdot,u_k)$ and $ \psi(\cdot,u_k)$ are solutions to 
\eqref{eq:Riccati1}. 
\end{corollary}

\appendix
\section{Semimartingale characteristics and martingales}
\label{auxiliary}

Let $(\Omega,\lijepof,(\lijepof_{t})_{t\in[0,T_*]},\P)$ denote a complete 
stochastic basis and $T_*$ denote a finite time horizon. Let $X$ be an 
$\erd$-valued semimartingale on this basis whose characteristics are absolutely 
continuous, i.e. its local characteristics are given by $(b,c,F;A)$ with 
$A_{t}=t$, for some truncation function $h$; cf. Jacod and Shiryaev
\citeyearpar[Prop.~II.2.9]{JacodShiryaev03}. Moreover, let $f:\er_{+}\times\erd\to\er$ 
be a function of class $C^{1,2}(\er_{+}\times\erd)$. 

The process $f(\cdot, X)$ is a real-valued semimartingale which has again 
absolutely continuous characteristics. Let us denote its local characteristics 
by $(b^{f}, c^{f}, F^{f})$ for a truncation function $h^{f}$. Then, noting that 
It\^o's formula holds for the function $f \in C^{1,2}(\er_{+} \times \erd)$ and 
reasoning as in the proof of Corollary A.6 from \citet{GollKallsen00}, we have 
that
\begin{align}\label{eq:local_char_f}
b^{f}_{t} &= \fpoint(t, X_{t-}) + \la \De f(t, X_{t-}), b_{t} \ra 
    + \frac{1}{2} \sum_{i,j=1}^{d} \De_{ij}^{2} f(t, X_{t-}) c_{t}^{ij} 
  \notag \\ \notag
& \quad  + \int_{\erd} \left( h^{f}\big(f(t, X_{t-}+x)-f(t, X_{t-})\big)
          - \la \De f(t, X_{t-}), h(x)\ra\right) F_{t}(\de x)\\
c^{f}_{t} &= \big\la \De f(t, X_{t-}), c_{t} \De f(t, X_{t-})\big\ra
 \phantom{\int}\\ \notag
F^{f}_t(G)  &= \int_{\erd} \indik_{G} \big(f(t, X_{t-}+x)-f(t, X_{t-})\big) 
	      F_{t}(\de x), \quad \quad G \in \lijepob(\er \setminus \{0\}).
\end{align}

\begin{proposition}\label{proplocmart}
Let $X$ be an $\erd$-valued semimartingale with absolutely continuous 
characteristics $(b,c,F)$ and let $f: \er_{+} \times \erd \to \er$ be a 
function of class $C^{1, 2}$ such that the process $Y$ defined by 
\begin{align}\label{def-exp-fun}
Y_{t} := \e^{f(t, X_{t})}
\end{align}
is exponentially special. If the following condition holds
\begin{align}\label{condlocmart}
\la \De f(t, X_{t-}), b_t \ra 
 &= - \fpoint(t, X_{t-})
  - \frac{1}{2} \sum_{i,j=1}^{d} \De_{ij}^{2} f(t,X_{t-}) c_{t}^{ij} \notag\\
 &\quad - \frac{1}{2} \big\la \De f(t, X_{t-}), c_t \De f(t, X_{t-})\big\ra 
 \phantom{\int_\er}\\
 &- \int_{\erd} \Big( \e^{f(t, X_{t-}+x)-f(t, X_{t-})} - 1 -  \la 
    \De f(t, X_{t-}), h(x)\ra\Big) F_{t}(\de x), \nonumber
\end{align}
then $Y$ is a local martingale.
\end{proposition}

\begin{proof}
The proof follows from Theorem 2.18 in \citet{KallsenShiryaev02}: set 
$\theta=1$ and apply the theorem to the semimartingale $f(\cdot, X)$. Indeed, 
since $f(\cdot, X)$ has absolutely continuous characteristics it is also 
quasi-left continuous, hence assertions (6) and (1) of Theorem 2.18. yield
$$
K^{f(\cdot, X)} (1)
  = \widetilde{K}^{f(\cdot, X)}(1) 
  = \int_{0}^{\cdot} \Big( b^{f}_{t} + \frac{1}{2} c^{f}_{t} 
    + \int_{\er} \big( \e^{x}-1-h^{f}(x)\big) F^{f}_{t}(\de x) \Big) \de t.
$$
By definition of the exponential compensator and Theorem 2.19 in Kallsen and Shiryaev
\citeyearpar{KallsenShiryaev02} it follows that
$$
\e^{f(\cdot, X)- K^{f(\cdot, X)}(1)} \in \lijepom_{\loc}.
$$
Therefore, $\e^{f(\cdot, X)} \in \lijepom_{\loc}$ if and only if $K^{f(\cdot, 
X)}(1)=0$ up to indistinguishability. Equivalently,
$$
b^{f}_{t} + \frac{1}{2} c^{f}_{t} 
  + \int_{\er} \big( \e^{x}-1-h^{f}(x)\big) F^{f}_{t}(\de x) = 0
$$
for every $t$. Inserting the expressions for $b^{f}, c^{f}$ and $F^{f}$, cf. 
\eqref{eq:local_char_f},  into the above equality yields condition 
\eqref{condlocmart}.
\end{proof}

\begin{proposition}\label{propUImart}
Let $X$ be an $\erd$-valued semimartingale with absolutely continuous 
characteristics $(b,c,F)$ such that
\begin{equation}\label{firstUI}
\int_{0}^{T_*} \int_{\erd} (|x|^{2} \wedge 1) F_{t}(\de x) \de t + 
\int_{0}^{T_*} \int_{|x|>1} |x| \e^{K |x|} F_{t}(\de x) \de t < C_{1}
\end{equation}
and
\begin{equation}\label{secondUI}
\int_{0}^{T_*} \| c_{t}\| \de t < C_{2},
\end{equation}
for some deterministic constants $C_{1},C_{2}>0$. Moreover, let $f: \er_{+} 
\times \erd \to \er$ be a function of class $C^{1, 2}$ and globally Lipschitz, 
i.e. there exists a constant $K>0$ such that
$$
|f(t, x)-f(t,y)| \leq K |x-y|, \qquad t \geq 0, \ \ x,y \in \erd.
$$
Then, the process $f(\cdot, X)$ is exponentially special, while the process $Y$ 
defined by \eqref{def-exp-fun} and satisfying \eqref{condlocmart} is a 
uniformly integrable martingale.
\end{proposition}

\begin{proof} 
The process $f(\cdot, X)$ is exponentially special if and only if
$$
\indik_{\{|x|>1\}} \e^{x} \ast \nu^{f} \in \lijepov.
$$
Hence, it suffices to show that $\indik_{\{|x|>1\}} \e^{x} \ast \nu^{f}_{T_*} < 
\infty$, as the integrand is positive. Since $f$ is globally Lipschitz, we have
\begin{align*}
\indik_{\{|x|>1\}} \e^{x} \ast \nu^{f}_{T_*} 
  &= \int_{0}^{T_*} \int_{|x|>1} \e^{x} F_t^{f} (\de x) \de t \\
  &\stackrel{\eqref{eq:local_char_f}}{=} 
     \int_{0}^{T_*} \int_{\erd} \indik_{\{|f(t, X_{t-}+x)-f(t, X_{t-})|>1\}} 
     \e^{f(t, X_{t-}+x) - f(t, X_{t-})} F_t (\de x) \de t \\
  &\leq \int_{0}^{T_*} \int_{K |x|>1} \e^{K |x|} F_t (\de x) \de t < \infty,
\end{align*}
which holds by the Lipschitz property and \eqref{firstUI}.

Moreover, if  $F=\e^{f(\cdot, X)} \in \lijepom_{\loc}$, using Proposition 3.4 
in \citet{Criens_Glau_Grbac_2015} it is also a uniformly integrable martingale 
if the following condition holds:
\begin{equation}\label{firUI}
\int_{0}^{T_*} \bigg( c^{f}_{t} + \int_{\erd} \Big[ (|x|^{2} \wedge 1) 
  + |x| \e^{x} \indik_{\{|x| > 1\}} \Big] F^{f}_t(\dx) \bigg) \dt < C^{f},
\end{equation}
for some constant $C^{f} >0$. We first check the condition for the diffusion 
coefficient
\begin{align*}
\int_{0}^{T_*} c^{f}_{t} \de t 
  &= \int_{0}^{T_*} \la \De f(t, X_{t-}), c_{t} \De f(t, X_{t-}) \ra \de t 
  \leq \int_{0}^{T_*} \|c_{t}\| |\De f(t, X_{t-})|^2 \de t < C_{1}^{f},
\end{align*}
which follows from \eqref{secondUI} and the fact that $\De f(\cdot, X_{-})$ is 
bounded as a consequence of $f$ being globally Lipschitz. As for the jump part, 
we have that
\begin{multline*}
\int_{0}^{T_*} \int_{\er} (|x|^{2} \wedge 1) F_t^{f} (\de x) \de t 
   + \int_{0}^{T_*} \int_{|x|>1} |x| \e^{x} F_t^{f} (\de x) \de t \\
  \stackrel{\eqref{eq:local_char_f}}{=}
    \int_{0}^{T_*} \int_{\erd} (|f(t, X_{t-}+x)-f(t, X_{t-})|^{2} \wedge 1) 
      F_t(\de x) \de t\\
  \quad+ \int_{0}^{T_*} \int_{\erd} \indik_{\{|f(t,X_{t-}+x)-f(t,X_{t-})|>1\}} 
    \\ \qquad\qquad \times |f(t, X_{t-}+x) - f(t, X_{t-})| \,
    \e^{f(t, X_{t-}+x) - f(t, X_{t-})} F_t(\de x) \de t\\
  \leq \int_{0}^{T_*} \int_{\erd} (K^{2} |x|^{2} \wedge 1) F(\de x) \de t 
    + \int_{0}^{T_*} \int_{K |x|> 1} K |x| \e^{K |x|} F^{f} (\de x) \de t 
  < C_{2}^{f},
\end{multline*}
using again the Lipschitz property and \eqref{firstUI}.
\end{proof}

Next, we provide the representation of $Y$ as a stochastic exponential.

\begin{lemma}
Let $X$ be an $\erd$-valued semimartingale with absolutely continuous 
characteristics $(B,C,\nu)$ and let $f:\er_{+} \times \erd \to \er_{+}$ be a 
function of class $C^{1,2} (\er_{+} \times \erd)$. Define a real-valued 
semimartingale $Y$ via \eqref{def-exp-fun}. If $Y \in \lijepom_{\loc}$, then it 
can be written as
$$
Y = \lijepoe \big( \De f(\cdot, X_{-}) \cdot X^{c} 
		  + W(\cdot, x) \ast (\mu^{X} - \nu) \big),
$$
where $X^{c}$ is the continuous martingale part of $X$, $\mu^{X}$ is the random 
measure of jumps of $X$ with compensator $\nu$ and
$$
W(\cdot, x):= \e^{f(\cdot, X_{-} +x) - f(\cdot, X_{-})} -1.
$$
\end{lemma}

\begin{proof}
Theorem 2.19 in \citet{KallsenShiryaev02} yields that
$$
 \e^{f(\cdot, X)} = \lijepoe \left( f(\cdot, X)^{c} + (\e^{x}-1) \ast 
(\mu^{f}-\nu^{f})\right),
$$
using that $f(\cdot, X)$ is quasi-left continuous since $X$ is also quasi-left 
continuous. Here $f(\cdot, X)^{c}$ denotes the continuous martingale part of 
$f(\cdot, X)$ and $\mu^{f}$ its random measure of jumps. The result now follows 
using the form of the local characteristics $c^f, F^f$ of the process 
$f(\cdot,X)$ in \eqref{eq:local_char_f}; see also the proof of Corollary A.6 in 
\citet{GollKallsen00}.
\end{proof}

\begin{lemma}\label{SLgirsanov}
Let $X$ be an $\erd$-valued semimartingale with absolutely continuous 
characteristics $(b,c,F)$ with respect to the truncation function $h$. Let 
$f:\er_{+} \times \erd \to \er_{+}$ be a function of class $C^{1,2}(\er_{+} 
\times \erd)$ and globally Lipschitz. Assume that conditions 
\eqref{condlocmart}, \eqref{firstUI} and \eqref{secondUI} are satisfied.

Define the probability measure $\P'\sim\P$ via
$$
\frac{\de \P'}{\de \P}\Big|_{\lijepof_\cdot} := \e^{f(\cdot, X)}\,.
$$
Then, the $\P'$-characteristics of the semimartingale $X$ are absolutely 
continuous and provided by $(b',c',F')$, where
\begin{align*}
b'_t &= b_t + c_t \beta_t + \int_{\erd} ( Y_t(x) - 1  ) h(x) F_t(\de x) \\
c'_t &= c_t \\
F'_t(\de x) &= Y_t(x) F_t(\de x),\phantom{\int^0}
\end{align*}
with $\beta_t = \De f(t,X_{t-})$ and $Y_t(x)= \e^{f(t,X_{t-}+x) - 
f(t,X_{t-})}$, for $t\in \er_{+}$ and $x\in\erd$.
\end{lemma}

\begin{proof}
The result follows directly from the previous lemma and Proposition 2.6 in 
\citet{Kallsen06}.
\end{proof}

\bibliographystyle{plainnat}
\bibliography{references}

\begin{thebibliography}{33}
\providecommand{\natexlab}[1]{#1}
\providecommand{\url}[1]{\texttt{#1}}
\expandafter\ifx\csname urlstyle\endcsname\relax
  \providecommand{\doi}[1]{doi: #1}\else
  \providecommand{\doi}{doi: \begingroup \urlstyle{rm}\Url}\fi

\bibitem[Andersen and Piterbarg(2010)]{Andersen_Piterbarg_2010}
L.~B.~G. Andersen and V.~V. Piterbarg.
\newblock \emph{Interest Rate Modeling, 3 Vols.}
\newblock Atlantic Financial Press, 2010.

\bibitem[Belomestny et~al.(2009)Belomestny, Mathew, and
  Schoenmakers]{BelomestnyMathewSchoenmakers06}
D.~Belomestny, S.~Mathew, and J.~Schoenmakers.
\newblock {Multiple stochastic volatility extension of the {LIBOR} market model
  and its implementation}.
\newblock \emph{Monte Carlo Methods Appl.}, 15:\penalty0 285--310, 2009.

\bibitem[Bianchetti and Morini(2013)]{BianchettiMorini13}
M.~Bianchetti and M.~Morini, editors.
\newblock \emph{Interest Rate Modelling After the Financial Crisis}.
\newblock Risk Books, 2013.

\bibitem[Brace et~al.(1997)Brace, G\c{a}tarek, and
  Musiela]{BraceGatarekMusiela97}
A.~Brace, D.~G\c{a}tarek, and M.~Musiela.
\newblock {The market model of interest rate dynamics}.
\newblock \emph{Math. Finance}, 7:\penalty0 127--155, 1997.

\bibitem[Brigo and Mercurio(2006)]{BrigoMercurio06}
D.~Brigo and F.~Mercurio.
\newblock \emph{{Interest Rate Models: Theory and Practice}}.
\newblock Springer, 2nd edition, 2006.

\bibitem[Criens et~al.(2015)Criens, Glau, and Grbac]{Criens_Glau_Grbac_2015}
D.~Criens, K.~Glau, and Z.~Grbac.
\newblock Martingale property of exponential semimartingales: a note on
  explicit conditions and applications to asset price and {LIBOR} models.
\newblock Preprint, \texttt{arXiv:1506.08127}, 2015.

\bibitem[Duffie et~al.(2003)Duffie, Filipovi{\'c}, and
  Schachermayer]{DuffieFilipovicSchachermayer03}
D.~Duffie, D.~Filipovi{\'c}, and W.~Schachermayer.
\newblock {Affine processes and applications in finance}.
\newblock \emph{Ann. Appl. Probab.}, 13:\penalty0 984--1053, 2003.

\bibitem[Eberlein and Kluge(2006{\natexlab{a}})]{EberleinKluge04}
E.~Eberlein and W.~Kluge.
\newblock {Exact pricing formulae for caps and swaptions in a {L}{\'e}vy term
  structure model}.
\newblock \emph{J. Comput. Finance}, 9\penalty0 (2):\penalty0 99--125,
  2006{\natexlab{a}}.

\bibitem[Eberlein and Kluge(2006{\natexlab{b}})]{EberleinKluge05}
E.~Eberlein and W.~Kluge.
\newblock {Valuation of floating range notes in L{\'e}vy term structure
  models}.
\newblock \emph{Math. Finance}, 16:\penalty0 237--254, 2006{\natexlab{b}}.

\bibitem[Eberlein and {\"O}zkan(2005)]{EberleinOezkan05}
E.~Eberlein and F.~{\"O}zkan.
\newblock {The L{\'e}vy LIBOR model}.
\newblock \emph{Finance Stoch.}, 9:\penalty0 327--348, 2005.

\bibitem[Filipovi{\'c} and Trolle(2013)]{Filipovic_Trolle_2011}
D.~Filipovi{\'c} and A.~Trolle.
\newblock {The term structure of interbank risk}.
\newblock \emph{J. Financ. Econom.}, 109:\penalty0 707--733, 2013.

\bibitem[Glasserman(2003)]{Glasserman03}
P.~Glasserman.
\newblock \emph{{Monte Carlo Methods in Financial Engineering}}.
\newblock Springer, 2003.

\bibitem[Goll and Kallsen(2000)]{GollKallsen00}
T.~Goll and J.~Kallsen.
\newblock {Optimal portfolios for logarithmic utility}.
\newblock \emph{Stochastic Process. Appl.}, 89:\penalty0 31--48, 2000.

\bibitem[Grbac and Runggaldier(2015)]{GrbacRunggaldier15}
Z.~Grbac and W.~J. Runggaldier.
\newblock \emph{Interest Rate Modeling: Post-Crisis Challenges and Approaches}.
\newblock \textit{SpringerBriefs in Quantitative Finance}, Springer, 2015.

\bibitem[Jacod and Shiryaev(2003)]{JacodShiryaev03}
J.~Jacod and A.~N. Shiryaev.
\newblock \emph{{Limit Theorems for Stochastic Processes}}.
\newblock Springer, 2nd edition, 2003.

\bibitem[Jamshidian(1997)]{Jamshidian97}
F.~Jamshidian.
\newblock {LIBOR and swap market models and measures}.
\newblock \emph{Finance Stoch.}, 1:\penalty0 293--330, 1997.

\bibitem[Jamshidian(1999)]{Jamshidian99}
F.~Jamshidian.
\newblock {{LIBOR} market model with semimartingales}.
\newblock Working Paper, NetAnalytic Ltd., 1999.

\bibitem[Kallsen(2006)]{Kallsen06}
J.~Kallsen.
\newblock {A didactic note on affine stochastic volatility models}.
\newblock In {Yu.} Kabanov, R.~Lipster, and J.~Stoyanov, editors, \emph{{From
  Stochastic Calculus to Mathematical Finance: The Shiryaev Festschrift}},
  pages 343--368. Springer, 2006.

\bibitem[Kallsen and Shiryaev(2002{\natexlab{a}})]{KallsenShiryaev02}
J.~Kallsen and A.~N. Shiryaev.
\newblock {The cumulant process and Esscher's change of measure}.
\newblock \emph{Finance Stoch.}, 6:\penalty0 397--428, 2002{\natexlab{a}}.

\bibitem[Kallsen and Shiryaev(2002{\natexlab{b}})]{KallsenShiryaev02b}
J.~Kallsen and A.~N. Shiryaev.
\newblock {Time change representation of stochastic integrals}.
\newblock \emph{Theory Probab. Appl.}, 46:\penalty0 522--528,
  2002{\natexlab{b}}.

\bibitem[Keller-Ressel et~al.(2013)Keller-Ressel, Papapantoleon, and
  Teichmann]{KellerResselPapapantoleonTeichmann09}
M.~Keller-Ressel, A.~Papapantoleon, and J.~Teichmann.
\newblock {The affine {LIBOR} models}.
\newblock \emph{Math. Finance}, 23:\penalty0 627--658, 2013.

\bibitem[Klein et~al.(2015)Klein, Schmidt, and
  Teichmann]{Klein_Schmidt_Teichmann_2015}
I.~Klein, T.~Schmidt, and J.~Teichmann.
\newblock No-arbitrage theory for bond markets.
\newblock Preprint, 2015.

\bibitem[Kluge(2005)]{Kluge05}
W.~Kluge.
\newblock \emph{{Time-inhomogeneous L{\'e}vy processes in interest rate and
  credit risk models}}.
\newblock PhD thesis, Univ. Freiburg, 2005.

\bibitem[Kluge and Papapantoleon(2009)]{KlugePapapantoleon06}
W.~Kluge and A.~Papapantoleon.
\newblock {On the valuation of compositions in {L}{\'e}vy term structure
  models}.
\newblock \emph{Quant. Finance}, 9:\penalty0 951--959, 2009.

\bibitem[Ladkau et~al.(2013)Ladkau, Schoenmakers, and
  Zhang]{Ladkau_Schoenmakers_Zhang_2013}
M.~Ladkau, J.~Schoenmakers, and J.~Zhang.
\newblock Libor model with expiry-wise stochastic volatility and displacement.
\newblock \emph{Int. J. Portfolio Analysis and Management}, 1:\penalty0
  224--249, 2013.

\bibitem[Miltersen et~al.(1997)Miltersen, Sandmann, and
  Sondermann]{MiltersenSandmannSondermann97}
K.~R. Miltersen, K.~Sandmann, and D.~Sondermann.
\newblock {Closed form solutions for term structure derivatives with log-normal
  interest rates}.
\newblock \emph{J. Finance}, 52:\penalty0 409--430, 1997.

\bibitem[Musiela and Rutkowski(1997)]{MusielaRutkowski97b}
M.~Musiela and M.~Rutkowski.
\newblock {Continuous-time term structure models: forward measure approach}.
\newblock \emph{Finance Stoch.}, 1:\penalty0 261--291, 1997.

\bibitem[Musiela and Rutkowski(2005)]{MusielaRutkowski05}
M.~Musiela and M.~Rutkowski.
\newblock \emph{{Martingale Methods in Financial Modelling}}.
\newblock Springer, 2nd edition, 2005.

\bibitem[Papapantoleon(2010)]{Papapantoleon10b}
A.~Papapantoleon.
\newblock {Old and new approaches to {LIBOR} modeling}.
\newblock \emph{Stat. Neerlandica}, 64:\penalty0 257--275, 2010.

\bibitem[Papapantoleon et~al.(2012)Papapantoleon, Schoenmakers, and
  Skovmand]{PapapantoleonSchoenmakersSkovmand10}
A.~Papapantoleon, J.~Schoenmakers, and D.~Skovmand.
\newblock {Efficient and accurate log-{L}{\'e}vy approximations to
  {L}{\'e}vy-driven {LIBOR} models}.
\newblock \emph{J. Comput. Finance}, 15\penalty0 (4):\penalty0 3--44, 2012.

\bibitem[Schoenmakers(2005)]{Schoenmakers05}
J.~Schoenmakers.
\newblock \emph{{Robust {LIBOR} Modelling and Pricing of Derivative Products}}.
\newblock Chapman \& Hall/CRC Press, 2005.

\bibitem[Wheatley(2012)]{Wheatley2012}
M.~Wheatley.
\newblock The {W}heatley review of {LIBOR}.
\newblock Technical report, HM Treasury, 2012.

\bibitem[Wu and Zhang(2006)]{Wu_Zhang_2006}
L.~Wu and F.~Zhang.
\newblock {{LIBOR} market model with stochastic volatility}.
\newblock \emph{J. Industr. Manag. Optim.}, 2:\penalty0 199--227, 2006.

\end{thebibliography}

\end{document}